\documentclass[aps,reprint,superscriptaddress,pra,10pt,floatfix]{revtex4-1}
\usepackage[utf8]{inputenc}
\usepackage[T1]{fontenc}
\usepackage{physics}
\usepackage{amsmath}
\usepackage{amssymb}
\usepackage{amsthm}
\usepackage{tikz}
\usepackage[breaklinks=true]{hyperref}
\usepackage{lipsum}
\usepackage{graphicx}
\usepackage{mathdots}
\usepackage{float}
\makeatletter
\let\newfloat\newfloat@ltx
\makeatother
\usepackage{algpseudocode}
\usepackage{algorithm}
\usepackage{dsfont}
\usepackage{tensor}

\theoremstyle{definition}
\newtheorem{definition}{Definition}[section]
\newtheorem{theorem}{Theorem}[section]
\newtheorem{proposition}{Proposition}[section]
\newtheorem{lemma}{Lemma}[theorem]
\newtheorem{corollary}{Corollary}[theorem]

\newtheorem*{example}{Example}
\usetikzlibrary{decorations.markings,knots}
\usetikzlibrary{decorations.pathreplacing,calligraphy}
\usetikzlibrary{shapes.geometric}
\usetikzlibrary{shapes,arrows}
\pgfdeclarelayer{background}
\pgfdeclarelayer{foreground}
\pgfsetlayers{background,main,foreground}

\tikzstyle{arrow}=[-,draw=black,postaction={decorate},decoration={markings,mark=at position 0.99 with {\arrow{>}}},thick]
\begin{document}

\title{Grand Unification of continuous-variable codes}
\author{Allan D. C. Tosta}
\affiliation{Federal University of Rio de Janeiro, Caixa Postal 68528, Rio de Janeiro, RJ 21941-972, Brazil}
\author{Thiago  O. Maciel}
\affiliation{Federal University of Rio de Janeiro, Caixa Postal 68528, Rio de Janeiro, RJ 21941-972, Brazil}
\author{Leandro Aolita}
\affiliation{Federal University of Rio de Janeiro, Caixa Postal 68528, Rio de Janeiro, RJ 21941-972, Brazil}
\affiliation{Quantum Research Center, Technology Innovation Institute, Abu Dhabi, UAE}
\begin{abstract}
Quantum error correction codes in continuous variables (also called CV codes, or single-mode bosonic codes) have recently been identified to be a technologically viable option for building fault-tolerant quantum computers. The best known examples are the GKP code and the cat-code, both of which were shown to have some advantageous properties over any discrete-variable, or qubit codes. It was recently shown that the cat-code, as well as other kinds of CV codes, belong to set of codes with common properties called rotation-symmetric codes. We expand this result by giving a general description of sets of codes with common properties, and rules by which they can be mapped into one another, effectively creating a unified description of continuous-variable codes. We prove that the properties of all of these sets of codes can be obtained from the properties of the GKP code. We also show explicitly how this construction works in the case of rotation-symmetric codes, re-deriving known properties and finding new ones.
\end{abstract}

\maketitle

\section{Introduction}
Quantum error correction codes defined over an infinite dimensional Hilbert space, such as a single-mode bosonic system, are called \textit{continuous variable codes} or \textit{bosonic codes}. It is known that bosonic codes have some advantages over codes embedded in finite-dimensional Hilbert spaces (e.g qubit-based codes), be it with regards to how it protects quantum information \cite{albert_performance_2018,guillaud_repetition_2019} or with the ways their codewords can be created and manipulated experimentally \cite{mirrahimi_dynamically_2014}. In particular, there have been recent proposals \cite{chamberland_building_2022,noh_low_2021,bourassa_blueprint_2021} for using bosonic codes to create logical qubits that can latter be used as a basis to define a qubit-based code over (code concatenation \cite{noh_fault-tolerant_2020}).

The first example of bosonic code was the Gottesman-Kitaev-Preskill code, or GKP code \cite{gottesman_encoding_2001}. It is defined as a subspace of the single-mode bosonic Fock space, spanned by codewords that are eigenstates of to two commuting finite translations, $T_{q}$ and $T_{p}$, over quadrature variables $q,p$ satisfying canonical commutation relations. In \cite{fukui_analog_2017} it was shown that the form of the GKP codewords imposes a probability distribution over syndrome measurements, in a way that is not possible for discrete-variable codes in general. This allows for the construction of better decoding algorithms, and for better behavior under concatenation with other qubit codes \cite{noh_fault-tolerant_2020,campagne-ibarcq_quantum_2020,bourassa_blueprint_2021}.

Another important example is the \textit{cat-code} \cite{leghtas_hardware-efficient_2013,mirrahimi_dynamically_2014}. Its codewords are linear combinations of coherent states with the same amplitude and different phases. This code can be implemented in a way such that bit-flip errors are exponentially less likely to occur, while phase-flip errors are linearly more likely \cite{lescanne_exponential_2020}. This error bias for cat-code qubits can be exploited for concatenating it with qubit-based codes tailored for correcting phase-flip errors, such as the repetition code \cite{guillaud_repetition_2019,chamberland_building_2022} and the $\mathrm{ZXXZ}$ surface code \cite{darmawan_practical_2021}.

However, despite these fast advancements there is no unified understanding of bosonic codes. First steps in this direction were taken by Grimsmo et. al. in \cite{grimsmo_quantum_2020}. There, they characterized a family of codes called \textit{rotation-symmetric codes}, which are codes whose codewords are eigenstates of a phase-space rotation by a fixed angle. It was shown that several codes, including cat-codes, belong to this family. This generalization also allowed for the design of general error correction protocols that apply to all of them.

Here we develop a framework that provides a grand unification of the understanding about a very large class of CV codes we call \textit{unitary-invariant codes} (or $U$-invariant codes, for a unitary $U$). These families of codes includes all known CV codes, such as the the GKP code, the cat-code, and other rotation-symmetric codes, as special cases. Essentially, we prove a theorem that can be used as a toolbox to constructively build an explicit map between any two families of unitary-invariant codes under some weak assumptions. We use this toolbox to show that all properties of a particular code family can be mapped to any other family. With this we build the set of logical operators and detectable errors of any family of unitary-invariant codes. As an application, we use this framework to map the family of rotation-invariant codes to the family of translation-invariant codes, the one containing the GKP code, and obtain all logical operators and detectable errors for all codes in both families.

In technical terms, the theorem we prove is based of methods developed in the 70's and 80's for creating a consistent theory of canonical transformations in quantum mechanics \cite{mello_nonlinear_1975,moshinsky_canonical_1978,moshinsky_canonical_1979,deenen_canonical_1980,plebanski_partial_1980,newton_quantum_1980,bauer_time_1983}. We use those technical tools to derive a rigorous, constructive method to embed an hermitian operator, whose spectrum satisfies a few conditions, into another operator whose spectrum in the whole real line. This allows us describe a map between those hermitian operators and the momentum ( position) operator as product of a unitary and a projection map. However, we prove that this construction also requires the tools of rigged Hilbert spaces to provide meaningful answers. Therefore, we also introduce a technical innovation by extending the formalism developed most rigorously in \cite{plebanski_partial_1980}, into rigged Hilbert spaces.

This paper is structured as follows: In Sec.\ref{sec:prel}, we do a brief review of rotation-symmetric codes (Sec.\ref{sub:rot-sym}), and of the GKP code (Sec.\ref{sub:gkp}). Next, in Sec.\ref{sec:resu} we state the definition of $U$-invariant codes, and use it to define \textit{translation-invariant codes} as well as some other examples (Sec.\ref{sub:u-inv-cod}). Then, we give a simplified recipe for how to map any unitary-invariant code into a generalized version of a translation-invariant code (Sec.\ref{sub:uni-bos}), and give the explicit map between rotation-symmetric and translation-symmetric codes as an example (Sec.\ref{sub:map-trans-rot}). Finally, in Sec.\ref{sec:meth} we briefly describe the mathematics behind our constructions. We contextualize our formalism historically, by discussing its birth as a tool to represent the canonical transformations of Hamiltonian mechanics in quantum theory. Next, we introduce the mathematics of partial isometries and explain how they are used to build maps between operators with different spectra that preserve algebraic relations.

\section{Preliminaries}\label{sec:prel}

Here, we begin our review. In Sec.\ref{sub:rot-sym}, we define what are rotation-symmetric codes, show that they are parametrized by a set of states in Fock space, and describe their set of detectable errors. In Sec.\ref{sub:gkp}, we define GKP codes, show their set of logical operators, their set of detectable errors and, compare these sets with the ones for rotation-symmetric codes. 

Let us introduce the following notation for eigenstates of an hermitian operator: for any hermitian operator $F$ and any of its eigenvalues $\lambda$, the corresponding eigenvector is denoted by $\ket{\lambda}_{F}$. Or in other words,
\begin{equation}
G\ket{\lambda}_{G}:=\lambda\ket{\lambda}_{G}.
\end{equation}
In a similar way, the "eigenbras" have a label to their left, as in $\tensor[_G]{\bra{\lambda}}{}$ The "ket-bras" have either the form $\tensor[_X]{\ketbra{\psi}{\psi'}}{_Y}$, if $\psi$ and $\psi'$ belong to Hilbert spaces defined by different operators $X$ and $Y$, or the form $\ketbra{\psi}{\psi'}_{X}$ if the Hilbert spaces are defined by the same operator.

We also introduce the following notation for the exponential of operators: for any hermitian operator $G$, and any operator $U(s)=e^{isG}$, we write
\begin{equation}
U_{s}:=U(s).
\end{equation}
For example, if $G=n$, where $n$ is as before, we write the rotation operator $R(\theta)$ as $R_{\theta}$. Or, if $G=-p$, we write the translation in position $T^{q}(\eta)=e^{-i\eta p}$ as $T^{q}_{\eta}$.

\subsection{Rotation-symmetric codes}\label{sub:rot-sym}

A \textit{rotation-symmetric code of order $N$} (or an $\mathcal{S}_{R_{2\pi/N}}$ code) is a two-dimensional subspace of the single-mode bosonic Fock space $\mathcal{F}$ defined by choosing a basis \{$\ket{i}_{\mathcal{S}}|i=0,1\}$, where
\begin{equation}
R_{\pi/N}\ket{i}_{\mathcal{S}_{R_{2\pi/N}}}=(-1)^{i}\ket{i}_{\mathcal{S}_{R_{2\pi/N}}},
\end{equation} 
with being $N$ is a non-zero integer. Note that any such basis is invariant under the rotation $R_{2\pi/N}$, giving the code family its name. This set of codes was first defined in \cite{grimsmo_quantum_2020}, and includes popular codes, such as cat-codes \cite{leghtas_hardware-efficient_2013,mirrahimi_dynamically_2014} and binomial codes \cite{michael_new_2016}.

Grimsmo et al. showed that the set of all such codes is parametrized by vectors called \textit{primitive states}, $\ket{\Theta}$, from which the associated logical states $\left\{\ket{j}^{\Theta}_{\mathcal{S}_{R_{2\pi/N}}}|j=0,1\right\}$ are obtained by
\begin{subequations}
\begin{align}
\ket{j}^{\Theta}_{\mathcal{S}_{R_{2\pi/N}}}&=\frac{\Pi^{j}_{R_{\pi/N}}\ket{\Theta}}{\sqrt{\mel{\Theta}{\Pi^{j}_{R_{\pi/N}}}{\Theta}}}\text{, with}\label{eq:rot_code_prim}\\
\Pi^{j}_{R_{\pi/N}}&=\sum_{m=0}^{2N-1}((-1)^{j}R_{\pi/N})^{m}\label{eq:rot_code_proj}
\end{align}
\end{subequations}
being orthogonal projectors. The condition that an arbitrary state must satisfy to be a primitive state of some $\mathcal{S}^{\Theta}_{R_{2\pi/N}}$ code is to have in its number-basis expansion at least two component states of the form $\ket{kN}_{n}$, one with $k$ even and one with $k$ odd.

For any $\mathcal{S}_{R_{2\pi/N}}$ code, the set of \textit{detectable errors} includes \textit{number-shift operators}
\begin{equation}
\Gamma_{m}=(a\sqrt{n})^{m}=\sum_{l=0}^{\infty}\ketbra{l}{l+m}_{n}
\end{equation}
for $m\in\{1,\dots,N\}$, as well as their hermitian conjugates. Also, for any $\mathcal{S}_{R_{2\pi/N}}$ code, the set of logical operators includes
\begin{equation}\label{eq:log-rot-sym}
\Bar{Z}=e^{i\frac{\pi}{N}n},\quad\Bar{S}=e^{i\frac{\pi}{2N^{2}}n^{2}}\text{, and }\Bar{T}=e^{i\frac{\pi}{4N^{4}}n^{4}},
\end{equation}
where $\Bar{Z}$, $\Bar{S}$ and $\Bar{T}$ are, respectively, the Pauli-Z gate, the phase gate and the $\pi/8$ gate\footnote{This form of the operators, and the fact that they can not be simpler, is proven by calculating their actions over $\ket{(2kN)_{n}}$ and $\ket{((2k+1)N)_{n}}$.}. Similarly, for any encoded two-qubit system of the form $\mathcal{S}_{R_{2\pi/N}}\otimes\mathcal{S}_{R_{2\pi/M}}$, where $N$ can be different than $M$, the set of two-qubit logical operators includes the gate
\begin{equation}
\Bar{CZ}=\mathrm{CROT}_{NM}:=e^{i(\pi/NM)n\otimes n},
\end{equation}
where $\Bar{CZ}$ is the controlled Pauli-Z gate. 

In \cite{grimsmo_quantum_2020}, it was shown that, if we are able to prepare the states $\ket{\pm_{R_{\pi/N}}}$ and measure in this basis, then the encoded gates above allow universal quantum computing, regardless of the specific $\mathcal{S}_{R_{2\pi/N}}$ code used (code-independent universality). They have also developed error correction protocols that are code-independent, which can be used in concatenation with other codes for producing fault-tolerant computation schemes.

Extending the class of primitive states to include non-normalizable states, it was shown that the primitive state
\begin{equation}
\ket{*}=\sum_{m\in\mathbb{N}}\ket{m}_{n}
\end{equation}
defines a particular code, called the \textit{ideal rotation-symmetric code of order} $N$, or $\mathcal{S}^{*}_{R_{2\pi/N}}$. This code has logical states of the form
\begin{equation}
\ket{i}^{*}_{\mathcal{S}_{R_{2\pi/N}}}\propto\sum_{k=0}^{\infty}\ket{(2k+i)N}_{n}.
\end{equation}

For $\mathcal{S}^{*}_{R_{2\pi/N}}$, it was shown that the set of logical operators includes the operators discussed above and the Pauli-X $\Bar{X}=(a\sqrt{n})^N$. The set of detectable errors however, is completely known, and is generated by linear combinations of products of
\begin{subequations}
\begin{align}
(a\sqrt{n})^{m}&\text{, where }m\in\{1,\dots,N-1\}\text{, and}\\
R_{\theta}&\text{, where }\theta\in\left(0,\frac{\pi}{N}\right).
\end{align}
\end{subequations}

\subsection{GKP codes}\label{sub:gkp}

A \textit{GKP code of scaling $\lambda$} (or an $\mathcal{GKP}_{\lambda}$ code) is a two-dimensional subspace spanned by the logical states $\ket{0_{\mathcal{GKP}_{\lambda}}}$ and $\ket{1_{\mathcal{GKP}_{\lambda}}}$, defined by
\begin{subequations}
\begin{align}
T^{q}_{\lambda\sqrt{\pi}}\ket{i}_{\mathcal{GKP}_{\lambda}}&=(-1)^{i}\ket{i}_{\mathcal{GKP}_{\lambda}},\\
T^{p}_{\sqrt{\pi}/\lambda}\ket{i}_{\mathcal{GKP}_{\lambda}}&=\ket{i\oplus_{2}1}_{\mathcal{GKP}_{\lambda}},
\end{align}
\end{subequations}
where $\oplus_{2}$ is addition modulo $2$, and \begin{equation}
T^{q}_{\lambda\sqrt{\pi}}=e^{-i(\lambda\sqrt{\pi})p},\quad T^{p}_{\sqrt{\pi}/\lambda}=e^{i\frac{\sqrt{\pi}}{\lambda}q}
\end{equation}
are translation operators, $\lambda$ is an arbitrary real scaling factor, and $q$ and $p$ are the position and momentum operators. Clearly, $\mathcal{GKP}_{\lambda}$ is invariant by the actions of both $T^{q}_{2\lambda\sqrt{\pi}}$ and $T^{p}_{2\sqrt{\pi}/\lambda}$, which commute.

The codewords have infinite norm, and are expressed in the momentum operator basis by
\begin{equation}
\ket{i}_{\mathcal{GKP}_{\lambda}}\propto\sum_{k\in\mathbb{Z}}\ket{(2k+i)\frac{\sqrt{\pi}}{\lambda}}_{p}.
\end{equation}
It is known from \cite{gottesman_encoding_2001} that the set of detectable errors is given by linear combinations of products of the operators 
\begin{subequations}
\begin{align}
T^{q}_{\eta}&\text{, with }\eta\in(-\lambda\sqrt{\pi},\lambda\sqrt{\pi})\text{, and }\\
T^{p}_{\zeta}&\text{, with }\zeta\in(-\sqrt{\pi}/\lambda,\sqrt{\pi}/\lambda).
\end{align}
\end{subequations}

It is also known \cite{noh_fault-tolerant_2020} that the single-qubit logical operators for $\mathcal{GKP}_{\lambda}$ are
\begin{subequations}
\begin{align}
\Bar{Z}=e^{-i\lambda\sqrt{\pi}p},\quad&\Bar{S}=e^{i\frac{\lambda^{2}}{2}p^{2}},\quad\Bar{T}=e^{i\frac{\lambda^{4}}{4\pi}\hat{p}^{4}}\label{eq:log-gkp}\\
\Bar{X}=e^{i(\sqrt{\pi}/\lambda)q}&\text{, and }\Bar{H}=e^{i\pi/2\left(\lambda^{2}p^{2}+\frac{q^{2}}{\lambda^{2}}\right)},
\end{align}
\end{subequations}
and the set of two-qubit logical operators for an encoded two-qubit system of the form $\mathcal{GKP}_{\lambda}\otimes\mathcal{GKP}_{\lambda'}$ includes the gate
\begin{equation}
\Bar{CZ}=e^{-i\lambda\lambda'p\otimes p}.
\end{equation}
Which together with the single-qubit gates, they form a universal gate set.

Looking at the form of the logical operators $\Bar{Z}$, $\Bar{S}$, and $\Bar{T}$ in Tab.\ref{tab:my_label}, it could suggest that an unitary map $U$ taking the operator $p$ into the operator $n$ would transform the $\mathcal{GKP}_{\lambda}$ code with $\lambda=-\sqrt{\pi}/N$ into the $\mathcal{S}^{*}_{R_{2\pi/N}}$ code. However, such an $U$ does not exist, since unitary maps are \textit{isospectral}, and $p$ and $n$ have different spectra. This is further illustrated by observing the last line in Tab.\ref{tab:my_label}, and noticing that no unitary transformation could map $(a\sqrt{n})^{N}$, which is itself non-unitary, into $e^{i(\sqrt{\pi}/\lambda)q}$, which is unitary. 

It is part of our objective in this work to show that these to codes are instead related by a composing a unitary map with a projection map.

\begin{table}[]
    \centering
    \begin{tabular}{|c||c|c|}
    \hline
         & $\mathcal{S}^{*}_{R_{2\pi/N}}$ & $\mathcal{GKP}_{\lambda}$\\
    \hline
    \hline
    $\Bar{Z}$ & $e^{i\frac{\pi}{N}n}$ & $e^{-i\lambda\sqrt{\pi}p}$ \\
    \hline
    $\Bar{S}$ & $e^{i\frac{\pi}{2N^{2}}n^{2}}$ & $e^{i\frac{\lambda^{2}}{2}p^{2}}$\\
    \hline
    $\Bar{T}$ & $e^{i\frac{\pi}{4N^{4}}n^{4}}$ & $e^{i\frac{\lambda^{4}}{4\pi}\hat{p}^{4}}$\\
    \hline
    $\Bar{X}$ & $(a\sqrt{n})^{N}$ & $e^{i(\sqrt{\pi}/\lambda)q}$\\
    \hline
    \end{tabular}
    \caption{Side-by-side comparison of single-qubit logical operators between the $\mathcal{S}^{*}_{R_{2\pi/N}}$ code and the $GKP_{\lambda}$ code.}
    \label{tab:my_label}
\end{table}

\section{Results}\label{sec:resu}

Here, give a succint description of our framework. In Sec.\ref{sub:u-inv-cod}, define \textit{unitary-invariant codes}, giving translation-symmetric codes as an example. In Sec.\ref{sub:uni-bos}, we provide a recipe for mapping any unitary-invariant code into a generalized version of the translation-symmetric codes. Finally, in Sec.\ref{sub:map-trans-rot} we give the mapping between translation and rotation symmetric codes as a use case.

\subsection{\textit{U}-invariant codes}\label{sub:u-inv-cod}

Here we define unitary-invariant code families and show that they are characterized by their codeword projectors, and a set of allowed primitive states. We also defined special elements of these families by specifying a particular primitive state. These special elements are essentially a generalization of the $\mathcal{GKP}_{\lambda}$ codes. These code families are the central object of our grand unification.

A \textit{unitary-invariant code} with unitary $U_{s}$ and parameter $s$ (or a $\mathcal{S}_{U_{s}}$ code) is a two-dimensional subspace of the Fock-space $\mathcal{F}$ defined by a basis $\{\ket{i}_{\mathcal{S}_{U_{s}}}|i=0,1\}$ that satisfies
\begin{equation}
U_{s/2}\ket{j}_{\mathcal{S}_{U_{s}}}=e^{i\frac{s}{2}G}\ket{j}_{\mathcal{S}_{U_{s}}}=(-1)^{j}\ket{j}_{\mathcal{S}_{U_{s}}},
\end{equation}
where $G$ is the hermitian generator of $U_{s}$. In other words, all unitary-invariant codes are stabilized by $U_{s}$ and have $U_{s/2}$ as the logical Pauli-Z. The prototypical example of family of unitary invariant codes is the family of rotation-symmetric codes $\mathcal{S}_{R_{2\pi/N}}$. 

Another particularly important example is the family of unitary-invariant codes with unitary $T^{q}_{2\lambda\sqrt{\pi}}$, which we call \textit{translation-symmetric codes of spacing $\lambda$} (or $\mathcal{S}_{T^{q}_{2\lambda\sqrt{\pi}}}$ codes). Just as we have Eqs.(\ref{eq:rot_code_prim},\ref{eq:rot_code_proj}) in the case of $\mathcal{S}_{R_{2\pi/N}}$ codes, the $\mathcal{S}_{T^{q}_{2\lambda\sqrt{\pi}}}$ codes are parametrized by a set of primitive states $\ket{\Theta}$ such that the logical states are given by
\begin{subequations}
\begin{align}
\ket{j}^{\Theta}_{\mathcal{S}_{T^{q}_{2\lambda\sqrt{\pi}}}}&=\frac{\Pi^{j}_{T^{q}_{\lambda\sqrt{\pi}}}\ket{\Theta}}{\sqrt{\mel{\Theta}{\Pi^{j}_{T^{q}_{\lambda\sqrt{\pi}}}}{\Theta}}}\text{, with}\\
\Pi^{j}_{T^{q}_{\lambda\sqrt{\pi}}}&=\sum_{m\in\mathbb{Z}}((-1)^{j}T^{q}_{\lambda\sqrt{\pi}})^{m},
\end{align}
\end{subequations}
where $\Pi^{j}_{T^{q}_{\lambda\sqrt{\pi}}}$ is a generalized projection operator. The condition an arbitrary state must satisfy to be a primitive state for some $\mathcal{S}_{T^{q}_{2\lambda\sqrt{\pi}}}$ code is to have in its momentum-basis expansion at least two component states of the form $\ket{k(\sqrt{\pi}/\lambda)}_{p}$, one with $k$ even and one with $k$ odd.

In general, given the spectrum of $G$, denoted by $\sigma(G)$, the projector that generate the codewords of any $\mathcal{S}_{U_{s}}$ are given by
\begin{equation}
\Pi^{j}_{U_{s}}=\sum_{m_{j}\in I^{j}(s)}\ketbra{(2m_{j}+j)\pi/s}{(2m_{j}+j)\pi/s}_{G},
\end{equation}
where $I^{j}(s)$ is the subset of $\mathbb{Z}$ that contains the numbers $m_{j}$ for which $(2m_{j}+j)(\pi/s)\in\sigma(G)$ with $j=0,1$. This implies that there are values of $s$ for which $\mathcal{S}_{U_{s}}$ may not be well defined. Using these projectors, we can parametrize $\mathcal{S}_{U_{s}}$ by a set of primitive vectors $\ket{\Theta}$ that only need to satisfy
\begin{equation}
\Pi^{j}_{U_{s}}\ket{\Theta}\neq0\text{, for }j=0,1.
\end{equation}

For the $\mathcal{S}_{T^{q}_{2\lambda\sqrt{\pi}}}$ codes, we can show that the primitive state
\begin{equation}
\ket{*}=\int_{-\infty}^{\infty}\dd{z}\ket{z}_{p},
\end{equation}
defines a particular code, which we call the \textit{ideal translation-symmetric code of spacing} $\lambda$ (or $\mathcal{S}^{*}_{T^{q}_{2\lambda\sqrt{\pi}}}$ code). This code has codewords given by
\begin{equation}
\ket{j}^{*}_{\mathcal{S}_{T^{q}_{2\lambda\sqrt{\pi}}}}\propto\sum_{k\in\mathbb{Z}}\ket{(2k+i)\frac{\sqrt{\pi}}{\lambda}}_{p},
\end{equation}
Which are exactly the same as the codewords of the $\mathcal{GKP}_{\lambda}$ code.

In general, we define the \textit{ideal code} $\mathcal{S}^{*}_{U_{s}}$ of an unitary-invariant code family as the code generated by the primitive state
\begin{equation}
\ket{*}=\sum_{\mu\in\sigma_{p}(G)}\ket{\mu}_{G}+\int_{\sigma_{c}(G)}\dd{\nu}\ket{\nu}_{G},
\end{equation}
where $\sigma_{p}(G)$ is the \textit{pure-point spectrum} of $G$, and $\sigma_{c}(G)$ is the \textit{absolutely continuous spectrum} of $G$ \cite{reed_functional_1980}. Therefore, unitary-invariant code families are characterized by their codeword projectors, and their set of allowed primitive states. We also defined special elements of these families by specifying a particular primitive state. These special elements are essentially a generalization of the $\mathcal{GKP}_{\lambda}$ codes. These code families are the central object of our grand unification.

\subsection{Unification of bosonic codes}\label{sub:uni-bos}

As stated earlier, our main result is a toolbox that allows us to map unitary-invariant codes among each other. This framework is essentially contained in the proof of the following theorem.

\begin{theorem}
If $A,B$ are non-degenerate, self-adjoint operators such that their spectra, $\sigma(A)$ and $\sigma(B)$, are a countable union of connected subsets of $\mathbb{R}$ then, there exists two pairs of maps $(\Xi_{A},\Omega_{A})$ and $(\Xi_{B},\Omega_{B})$ such that: $\Omega_{B}\circ\Xi_{A}$ and $\Omega_{A}\circ\Xi_{B}$ are algebra homomorphisms and 
\begin{equation}
(\Omega_{B}\circ\Xi_{A})[A]=B,\quad(\Omega_{A}\circ\Xi_{B})[B]=A.
\end{equation}
\end{theorem}

\begin{proof}
See Sec.\ref{app_sub:code-recip}.
\end{proof}

The most important consequence of this theorem is the following. Let $\mathcal{S}_{U_{s}}$ and $\mathcal{S}_{V_{s}}$ be two unitary-invariant code families with unitary stabilizers $U_{s}=e^{isA}$ and $V_{s}=e^{isB}$. Then, by Thm.\ref{thm:main}, we have that 
\begin{equation}
(\Omega_{B}\circ\Xi_{A})[U_{s}]=V_{s},\quad(\Omega_{A}\circ\Xi_{B})[V_{s}]=U_{s},
\end{equation}
for any $s\in\mathbb{R}$. This implies that the maps $\Omega_{B}\circ\Xi_{A}$ and $\Omega_{A}\circ\Xi_{B}$ send $\mathcal{S}_{U_{s}}$ codes into $\mathcal{S}_{V_{s}}$ codes and vice-versa. 

Our objective is to explicitly build the pairs of maps $(\Xi_{A},\Omega_{A})$ and $(\Xi_{B},\Omega_{B})$. We do this by finding maps taking the operator $A$ (or $B$) to the operator
\begin{equation}
\mathbf{M}_{\mathbb{R}}=\int_{\mathbb{R}}\dd{\nu}\nu\ketbra{\nu}{\nu}_{\mathbf{M}_{\mathbb{R}}},
\end{equation}
called the \textit{standard multiplication operator} over the set $\mathbb{R}$, and vice-versa. This mapping procedure is illustrated in Fig.[\ref{fig:unification}], and the mapping from $A$ to $\mathbf{M}_{\mathbb{R}}$ described by the informal algorithm Alg.[\ref{alg:recipe_code}].

\begin{algorithm}[ht]
\caption{Recipe for mapping $A$ to $\mathbf{M}_{\mathbb{R}}$}
\begin{algorithmic}
\Require Let $A$ be hermitian, non-degenerate and admissible
\State Find a family of operators $\{A_{\mu}\}_{\mu\in\Lambda}$ such that
\begin{itemize}
    \item The set $\{A_{\mu}\}_{\mu\in\Lambda}$ is the image set of a bijective function of $A$,
    \item $A\subset\{A_{\mu}\}_{\mu\in\Lambda}$,
    \item $\bigcup_{\mu\in\Lambda}\sigma(A_{\mu})=\mathbb{R}$, and
    \item $\bigcap_{\mu\in\Lambda}\sigma(A_{\mu})=\varnothing$.
\end{itemize}
\State Define $\mathcal{H}_{\Lambda}$ as the Hilbert space generated by the set of formal states $\{\ket{\mu}^{\Lambda}|\mu\in\Lambda\}$
\State Define $J_{A}$ by
\begin{equation}
J_{A}=\sum_{\mu\in\Lambda}A_{\mu}\otimes\ketbra{\mu}{\mu}^{\Lambda}
\end{equation}
\State Define the map $\iota_{A}$ by $\iota_{A}(\sum_{i=1}^{n}\alpha_{i}A^{i})=\sum_{i=1}^{n}\alpha_{i}J_{A}^{i}$ for all $n\in\mathbb{N}$.
\State Define $V_{\mathbf{M}_{\mathbb{R},A_{\mu}}}$ by
\begin{equation}
\displaystyle{V_{\mathbf{M}_{\mathbb{R}},A_{\mu}}=\sum_{\lambda\in\sigma_{p}(A_{\mu})}\ket{\lambda}_{\mathbf{M}_{\mathbb{R}}}\bra{\lambda}_{A_{\mu}}+\int_{\sigma_{c}(A_{\mu})}\dd{z}\ket{z}_{\mathbf{M}_{\mathbb{R}}}\bra{z}_{A_{\mu}}},
\end{equation}
\State Define $U_{\mathbf{M}_{\mathbb{R}},J_{A}}$ by
\begin{equation}
U_{\mathbf{M}_{\mathbb{R}},J_{A}}=\sum_{\mu\in\Lambda}V_{\mathbf{M}_{\mathbb{R}},A_{\mu}}\otimes\bra{\mu}^{\Lambda},
\end{equation}
\State Define $\Xi_{A}$ by $\Xi_{A}(O)=(U_{\mathbf{M}_{\mathbb{R}},J_{A}})\iota_{A}(O)(U^{\dagger}_{\mathbf{M}_{\mathbb{R}},J_{A}})$ for any operator $O$ . Then, we have that $\Xi_{A}(A)=\mathbb{M_{R}}$.
\end{algorithmic}
\label{alg:recipe_code}
\end{algorithm}

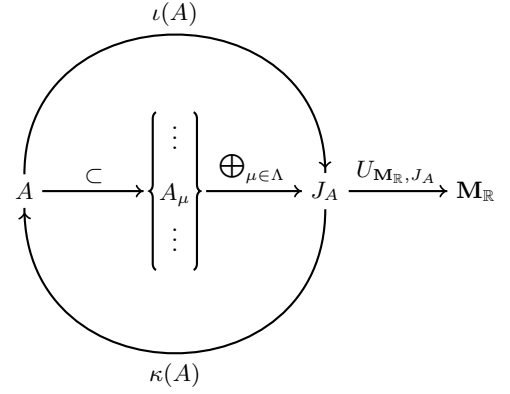
\begin{figure}[ht]
\centering
\begin{tikzpicture}
\node (c1) at (-2,2.75){};
\node (c2) at (2,2.75){};
\node (d1) at (-2,-2.75){};
\node (d2) at (2,-2.75){};
\node (a1) at (-2,0.05){$A$};
\node[anchor=south] (a2) at (0,0.5){$\vdots$};
\node (a3) at (0,0){$A_{\mu}$};
\node (a4) at (0,-0.5){$\vdots$};
\node (lbr) at (-0.3,0.05){};
\draw[thick,decorate,decoration = {brace}] (-0.25,-1) --  (-0.25,1.1);
\draw[thick,decorate,decoration = {brace}] (0.25,1.1) --  (0.25,-1);
\draw[thick,arrow] (a1) to (lbr);
\node (subset) at (-1.075,0.25){$\mathbf{\subset}$};
\node (ja) at (2,0.05){$J_{A}$};
\node (rbr) at (0.3,0.05){};
\draw[thick,arrow] (rbr) to (ja);
\node (dsum) at (1.025,0.35){$\bigoplus_{\mu\in \Lambda}$};
\draw[thick,arrow] (a1) .. controls (c1) and (c2) .. (ja);
\draw[thick,arrow] (ja) .. controls (d2) and (d1) .. (a1);
\node (ea) at (0,2.4){$\iota(A)$};
\node (pa) at (0,-2.4){$\kappa(A)$};
\node (Mr) at (4,0.05){$\mathbf{M}_{\mathbb{R}}$};
\draw[thick,arrow] (ja) to (Mr);
\node (U) at (2.95,0.3){$U_{\mathbf{M}_{\mathbb{R}},J_{A}}$};
\end{tikzpicture}
\caption{Illustration of Alg.\ref{alg:recipe_code}. First, we look for a set of operators $\{A_{\mu}\}_{\mu\in\Lambda}$ that includes $A$, then we build an operator $J_{A}$ by taking the direct sum over this family. If we choose a family such that the union of their spectra is the real line, we can build the unitary map $U_{\mathbf{M}_{\mathbb{R}},J_{A}}$, where $\mathbf{M}_{\mathbb{R}}$ is the standard multiplication operator. Then, the map $\iota_{A}$ takes $A$ to $J_{A}$, and any function $f(A)$ into $f(J_{A})$. In the end, the map we are after is $\Xi_{A}$, that by construction must satisfy $\Xi_{A}(A)=\mathbb{R}$. The map $\kappa_{A}$ is the standard projection, taking $J_{A}$ to $A$. From $\kappa_{A}$ we can define the map $\Omega_{A}$ by $\Omega_{A}(O):=\kappa((U^{\dagger}_{\mathbf{M}_{\mathbb{R}},J_{A}})O(U_{\mathbf{M}_{\mathbb{R}},J_{A}}))$, that satisfy $\Omega_{A}(\mathbf{M}_{\mathbb{R}})=A$}
\label{fig:unification}
\end{figure}

The operator $J_{A}$ is built as a direct sum over the family $\{A_{\mu}\}_{\mu\in\Lambda}$. Given this direct sum structure, the maps $\kappa_{A}$ and $\iota_{A}$ are the canonical projection and co-projection, respectively, over the $A$-th component of $J_{A}$. These maps end up being algebra homomorphisms, that is, $\iota_{A}$ sends a Taylor series of $A$ into the same Taylor series over $J_{A}$, and $\kappa_{A}$ does the reverse (see Sec.\ref{app_sub:code-recip}). The operators $V_{\mathbf{M}_{\mathbb{R}},J_{A}}$ are called \textit{partial isometries} and are discussed in Sec.\ref{sec:meth}. These are essentially a generalization of unitary maps.

Then, since $J_{A}$ is built to have the same spectrum as $\mathbf{M}_{\mathbb{R}}$, regardless of $A$, we can use the partial isometries to build the map $U^{\dagger}_{\mathbf{M}_{\mathbb{R}},J_{A}}$, which ends up being unitary. Therefore, we can send $\mathbf{M}_{\mathbb{R}}$ to $J_{A}$ via this unitary, or $J_{A}$ to $\mathbf{M}_{\mathbb{R}}$ via its adjoint. Using these ingredients, we can build an algebra homomorphism $\Xi_{A}$ that sends $A$ to $\mathbf{M}_{\mathbb{R}}$, and another algebra homomorphism $\Omega_{A}$ that sends $\mathbf{M}_{\mathbb{R}}$ to $A$. So if $B$ is another operator satisfying the conditions of Thm.\ref{thm:main}, we can apply Alg.[\ref{alg:recipe_code}] to $B$ to obtain maps $\Xi_{B}$, $\Omega_{B}$ that, by construction, must be algebra homomorphisms and, together with $\Omega_{A}$ and $\Xi_{A}$, must satisfy Eq.[\ref{eq:thm_main}].

We can also use these maps to map logical operators commuting with $U_{s}$ and $U_{s/2}$ into logical operators commuting with $V_{s}$ and $V_{s/2}$ having the same action over logical states. Essentially, if the code families $\mathcal{S}_{U_{s}}$ and $\mathcal{S}_{V_{s}}$ are well-defined, and $\Bar{O}_{U}$ is a logical operator that commutes with $U_{s/2}$ then the operator $\Bar{O}_{V}$, given by
\begin{equation}
\Bar{O}_{V}=\Omega_{B}\circ\Xi_{A}(\Bar{O}_{U}),
\end{equation}
is a logical operator for the corresponding code in the code family $\mathcal{S}_{V_{s}}$, that has the same action as $\Bar{O}_{U}$. This result can be extended to include non-diagonal single-qubit gates for codes where those are well-defined, as proven in Thm.\ref{thm:log-op}. To illustrate this algorithm and its applications, consider the following example.

\subsection{Mapping translation-symmetric codes into rotation-symmetric codes}\label{sub:map-trans-rot}

In this example, we take the standard multiplication operator to be the momentum operator $p$. Then, our main objective is to find the map $\Omega_{n}$
that takes $p$ to the number operator $n$. Following Alg.(\ref{alg:recipe_code}), we must find a family of operators $\{n_{\mu}\}_{\mu\in\Lambda}$ satisfying the recipe's conditions. 

\textit{\textbf{Applying the algorithm:}} Consider the operators
\begin{equation}
n_{\nu}:=n+\nu\mathds{1}_{\mathbb{N}}.
\end{equation}
If we take the index $\nu$ to belong to the interval $[0,1]$, the union of their spectra will be the non-negative reals $\mathbb{R}^{+}$, but their intersection will contain the set $\{1,2,...\}$. Therefore, to remove the intersection we exclude $1$ from the interval. Since we have found a set of operators with the union of their spectra given by the non-negative reals, now we need to include more operators, to obtain the negative reals. 

Consider the operators
\begin{equation}
n_{(\tau,\nu)}:=\tau(n+\nu\mathds{1}_{\mathbb{N}}).
\end{equation}
If we take $\tau$ to be either $1$ or $-1$, then the union of all $n_{(\tau,\nu)}$ is the whole real line, but its intersection is $\{0\}$, which is non-empty. This implies that we can not have neither $\tau=1$ and $\nu=1$, nor $\tau=-1$ and $\nu=0$ at the same time. With these restrictions, the family $\{n_{\mu}\}_{\mu\in\Lambda}$ has an index set of the form $\Lambda=\{-1\}\cross(0,1]\cup\{1\}\cross[0,1)$, and the union of their spectra is $\mathbb{R}$ with an empty intersection.

Now, we need to define the Hilbert space $\mathcal{H}_{\Lambda}$. In general, for any point $x\in\mathbb{R}$, we define the Hilbert space $\mathcal{H}_{x}$ as the space spanned by the basis vector $\ket{x}$, and for any interval $[a,b]$, we define the Hilbert space $\mathcal{H}_{[a,b]}$ as the space $L^{2}([a,b])$ of square-integrable functions. If $A,B$ are subsets of $\mathbb{R}$, define $\mathcal{H}_{A\cup B}=\mathcal{H}_{A}\oplus\mathcal{H}_{B}$, and $\mathcal{H}_{A\cross B}=\mathcal{H}_{A}\otimes\mathcal{H}_{B}$. 

Since $\Lambda=\{-1\}\cross(0,1]\cup\{1\}\cross[0,1)$, then
\begin{equation}
\mathcal{H}_{\Lambda}=\mathbb{C}\otimes(L^{2}((0,1])\oplus L^{2}([0,1))).
\end{equation}
However, if $X$ is any subset of $\mathbb{R}$ of Lebesgue measure $0$, then $L^{2}([a,b])$ is isomorphic to $L^{2}([a,b]-X)$. Therefore, $\mathcal{H}_{\Lambda}$ is isomorphic to $\mathbb{C}^{2}\otimes L^{2}([0,1])$ as a Hilbert space.

In order to define $J_{n}$ we need to say what summations over the set $\Lambda$ mean. We define such summations in terms of the structure of $\mathcal{H}_{\Lambda}$, which in this case leads to the substitution
\begin{equation}\label{eq:sum-rep}
\sum_{\mu\in\Lambda}\rightarrow\sum_{\tau\in\{-1,1\}}\int^{1}_{0}\dd{\nu},
\end{equation}
from which we can define 
\begin{equation}
J_{n}:=\sum_{\tau\in\{-1,1\}}\int^{1}_{0}\dd{\nu}n_{\tau,\nu}\otimes\ketbra{\tau}{\tau}\otimes\ketbra{\nu}{\nu}.
\end{equation}

Here we make an important observation. By the form of $J_{n}$, we can say that the Hilbert space it acts upon is given by $\mathcal{H}_{\mathbb{N}}\otimes\mathbb{C}^{2}\otimes L^{2}([0,1])$, where $\mathcal{H}_{\mathbb{N}}$ is the Fock space of a single mode. However, our construction ensures that $\sigma(J_{n})=\mathbb{R}$, and therefore we must have $\mathcal{H}_{\mathbb{N}}\otimes\mathbb{C}^{2}\otimes L^{2}([0,1])$ and $L^{2}(\mathbb{R})$ isomorphic as Hilbert spaces. 

To see this, consider the Hilbert spaces $H_{\sigma(n_{(\tau,\nu)})}$. Since the union of the sets $\sigma(n_{(\tau,\nu)})$ over $\Lambda$ is $\mathbb{R}$, we can build $\mathcal{H}_{\mathbb{R}}=L^{2}(\mathbb{R})$ by taking the direct sum
\begin{equation}
\bigoplus_{(\tau,\nu)\in\Lambda}\mathcal{H}_{\sigma(n_{(\tau,\nu)})}.
\end{equation}
For $\tau=1$, we have that $\mathcal{H}_{\sigma(n_{(1,\nu)})}=\mathcal{H}_{\mathbb{N}+\nu}$ for $\nu\in[0,1)$. The sets $\mathbb{N}+\nu$ have the same cardinality for every $\nu$, which implies that all $\mathcal{H}_{\mathbb{N}+\nu}$ are isomorphic to $\mathcal{H}_{\mathbb{N}}$. Therefore, we have that direct sum $\bigoplus_{\nu\in[0,1)}\mathcal{H}_{\mathbb{N}+\nu}$ is equal to $\mathcal{H}_{\mathbb{N}}\otimes L^{2}([0,1))$. Similarly, for $\tau=-1$ we have $\mathcal{H}_{n_{(-1,\nu)}}=\mathcal{H}_{\mathbb{N}-\nu}$ for $\nu\in(0,1]$. Applying the same argument as for the $\tau=1$ case, we have that $\mathcal{H}_{\mathbb{R}}$ is isomorphic to
\begin{equation}
\mathcal{H}_{\mathbb{N}}\otimes(L^{2}((0,1])\oplus L^{2}([0,1))),
\end{equation}
which is isomorphic to $\mathcal{H}_{\mathbb{N}}\otimes\mathbb{C}^{2}\otimes L^{2}([0,1])$. 

With this out of the way, and having found $J_{n}$, we are now in position to build $U_{p,J_{n}}$, the unitary that maps $p$ to $J_{n}$. To this end, we follow Alg.(\ref{alg:recipe_code}) and define the partial isometries 
\begin{equation}
V_{p,n_{(\tau,\nu)}}:=\sum_{\lambda\in(\tau(\mathbb{N}+\nu))}\ket{\lambda}_{p}\bra{\lambda}_{n_{(\tau,\nu)}}.
\end{equation}
This can also be written as
\begin{equation}
V_{p,n_{(\tau,\nu)}}=\sum_{m\in\mathbb{N}}\ket{\tau(m+\nu)}_{p}\bra{\tau(m+\nu)}_{n_{(\tau,\nu)}}.
\end{equation}
Now, since $n_{(\tau,\nu)}=\tau(n+\nu\mathds{1}_{\mathbb{N}})$, we can use the spectral theorem and write $\bra{\tau(m+\nu)}_{n_{(\tau,\nu)}}=\bra{m}_{n}$. 

Using a slight extension of the Dirac notation that is based on the theory of rigged Hilbert spaces, which we explain in Sec.\ref{sec:meth}, we can show that 
\begin{subequations}
\begin{align}
V^{\dagger}_{p,n_{(\tau,\nu)}}V_{p,n_{(\tau',\nu')}}&=\delta_{\tau,\tau'}\delta_{\nu,\nu'}\mathds{1}_{\tau(\mathbb{N}+\nu)}\text{, and}\label{eq:partial-p_n1}\\
V_{p,n_{(\tau,\nu)}}V^{\dagger}_{p,n_{(\tau',\nu')}}&=\delta_{\tau,\tau'}\delta_{\nu,\nu'}\sum_{\lambda\in\tau(\mathbb{N}+\nu)}\ketbra{\lambda}{\lambda}_{p}\label{eq:partial-p_n2}.
\end{align}
\end{subequations}

Therefore, using the definition of $U_{p,J_{n}}$ as described by Alg.(\ref{alg:recipe_code}) and the summation replacement defined in Eq.(\ref{eq:sum-rep}), we have that
\begin{equation}
U_{p,J_{n}}=\sum_{\tau\in\{-1,1\}}\int^{1}_{0}\dd{\nu}V_{p,n_{(\tau,\nu)}}\otimes(\bra{\tau}\otimes\bra{\nu}).
\end{equation}
Then, using Eqs.(\ref{eq:partial-p_n1},\ref{eq:partial-p_n2}) we can see that
\begin{equation}
U_{p,J_{n}}^{\dagger}U_{p,J_{n}}=U_{p,J_{n}}U^{\dagger}_{p,J_{n}}=\mathds{1}_{\mathbb{R}}.
\end{equation}
From the definition of $V_{p,n_{(\tau,\nu)}}$ can prove the identities
\begin{subequations}
\begin{align}
V^{\dagger}_{p,n_{(\tau,\nu)}}pV_{p,n_{(\tau,\nu)}}&=n_{(\tau,\nu)},\\
V_{p,n_{(\tau,\nu)}}n_{(\tau,\nu)}V^{\dagger}_{p,n_{(\tau,\nu)}}&=\sum_{\lambda\in(\tau\mathbb{N}+\nu)}\lambda\ketbra{\lambda}{\lambda}_{p},
\end{align}
\end{subequations}
which imply that 
\begin{equation}
(U_{p,J_{n}})J_{n}(U^{\dagger}_{p,J_{n}})=p,\quad (U_{p,J_{n}})p(U^{\dagger}_{p,J_{n}})=J_{n}.
\end{equation}

Finally, we define $\kappa_{n}$ as the map $\kappa_{n}[\cdot]=K_{n}[\cdot]K_{n}$ where we find the $K_{n}$ operator, by noticing that $n_{(1,0)}=n$, which implies that $K_{n}=\mathds{1}_{\mathbb{N}}\otimes\bra{1}\otimes\bra{0}$. Then, the $\Omega_{n}$ map, can be written as $\Omega_{n}[\cdot]=\Upsilon_{n}[\cdot]\Upsilon^{\dagger}_{n}$, where $\Upsilon_{n}=K_{n}U^{\dagger}_{p,J_{n}}$. With $\Omega_{n}$ we can now map any $\mathcal{S}_{T^{q}_{2\lambda\sqrt{\pi}}}$ code into some $\mathcal{S}_{R_{2\pi/N}}$ code. However, this mapping does not produce a valid $\mathcal{S}_{R_{2\pi/N}}$ code for all choices of $\lambda$ and $\ket{\Theta}$. 

For example, if we choose $\ket{\Theta}=\int_{\mathbb{R}}\dd{l}\ket{l}_{p}$ we see that
\begin{equation}
\Upsilon_{n}\int_{\mathbb{R}}\dd{l}\ket{l}_{p}=\sum_{m\in\mathbb{N}}\ket{m}_{n},
\end{equation}
which is the primitive state giving the ideal $\mathcal{S}_{R_{2\pi/N}}$. But we can also show that $\Upsilon_{n}T^{p}_{2\sqrt{\pi}/\lambda}{\Upsilon_{n}}^{\dagger}$ is non-zero only when $\lambda=k\sqrt{\pi}/l$ for some integers $k,l$. In particular, we have that
\begin{subequations}
\begin{align}
\Upsilon_{n}T^{q}_{2\pi/N}{\Upsilon_{n}}^{\dagger}&=R_{2\pi/N},\\
\Upsilon_{n}T^{p}_{2N}{\Upsilon_{n}}^{\dagger}&=(a\sqrt{n})^{2N},
\end{align}
\end{subequations}
which implies that $\Upsilon_{n}$ only gives a well defined mapping between $\mathcal{S}^{*}_{T^{q}_{2\lambda\sqrt{\pi}}}$ and $\mathcal{S}^{*}_{R_{2\pi/N}}$ when $\lambda=-\sqrt{\pi}/N$.

\textit{\textbf{Finding logical operators and detectable errors:}} To obtain the logical operators for $\mathcal{S}^{*}_{R_{2\pi/N}}$, we just need to apply $\Omega_{n}$ to the expressions for the logical operators for the $\mathcal{GKP}_{-\sqrt{\pi}/N}$ codes. We can see in Tab.\ref{tab:my_label_2} that for the operators $\Bar{Z}$, $\Bar{S}$, $\Bar{T}$ the calculations reproduce Tab.\ref{tab:my_label}. However, for the logical Pauli-X and the logical Hadarmard the calculation yields
\begin{subequations}
\begin{align}
\Bar{X}&=(a\sqrt{n})^{N}\text{, and}\\
\Bar{H}&=\frac{1}{\sqrt{2\pi}}\sum_{m,m'}e^{-i\pi\frac{mm'}{N^{2}}}\ketbra{m}{m'}_{n}.
\end{align}
\end{subequations}
Up to our knowledge, this is the first explicit form of the Hadamard gate for rotation-symmetric codes, which is result in itself.

To find the set of detectable errors, we just need apply $\Omega_{n}$ to the set of operators generated by $T^{q}_{\eta}$ with $\eta\in(-N,N)$ and $T^{p}_{\zeta}$ with $\zeta\in(-\pi/N,\pi/N)$. These calculations shows us that the values of $\eta$ for which the application of $\Omega_{n}$ to generators $T^{q}_{\eta}$ is non-trivial are exactly for the integers $-N+1,\dots,N-1$ excluding $0$, where the negative integers gives us the conjugate transpose of $(a\sqrt{N})^{l}$ for  $l=1,\dots,N-1$.

\begin{table}[]
    \centering
    \begin{tabular}{|c||c|c|}
    \hline
         & $\mathcal{S}^{*}_{R_{2\pi/N}}$ & $\mathcal{GKP}_{\lambda}$\\
    \hline
    \hline
    $\Bar{Z}$ & $e^{i\frac{\pi}{N}n}$ & $e^{-i\lambda\sqrt{\pi}p}$ \\
    \hline
    $\Bar{S}$ & $e^{i\frac{\pi}{2N^{2}}n^{2}}$ & $e^{i\frac{\lambda^{2}}{2}p^{2}}$\\
    \hline
    $\Bar{T}$ & $e^{i\frac{\pi}{4N^{4}}n^{4}}$ & $e^{i\frac{\lambda^{4}}{4\pi}\hat{p}^{4}}$\\
    \hline
    $\Bar{X}$ & $(a\sqrt{n})^{N}$ & $e^{i(\sqrt{\pi}/\lambda)q}$\\
    \hline
    $\Bar{H}$ & $\mathbf{\frac{1}{\sqrt{2\pi}}\sum_{m,m'}e^{-i\pi\frac{mm'}{N^{2}}}\ketbra{m}{m'}_{n}}$ & $e^{i\pi/2\left(\lambda^{2}p^{2}+\frac{q^{2}}{\lambda^{2}}\right)}$\\
    \hline
    \end{tabular}
    \caption{Side-by-side comparison of single-qubit logical operators between the $\mathcal{S}^{*}_{R_{2\pi/N}}$ code and the $GKP_{\lambda}$ code. The expression in boldface is the explicit form of the Hadamard gate, and to the best of our knowledge is not present in the literature.}
    \label{tab:my_label_2}
\end{table}

\section{Methods: Partial Isometries and Decompositions of maps}\label{sec:meth}

In this section we introduce the mathematical methods used by Alg.\ref{alg:recipe_code} to produce the correct mappings between unitary-invariant code families. We also give a brief historical overview of the its early developments, and how we have generalized them into the framework we present in this work.

The problem of mapping operators with different spectra has an interesting connection with the problem of finding the quantum version of classical canonical transformations \cite{mello_nonlinear_1975}. In general, canonical transformations between classical variables $(q,p)$ and $(q',p')$, for example
\begin{equation}
p'=\frac{p^{2}-q^{2}}{2},\quad q'=\ln{\abs{p+q}},
\end{equation}
can not be translated directly into transformations between operators for several reasons, such as the map not being bijective, as it is the case of this example.

However, since this transformation is canonical, the operator versions of both $(q,p)$ and $(q',p')$ must satisfy
\begin{equation}
[q,p]=i\mathds{1}\text{, and }[q',p']=i\mathds{1},
\end{equation}
where $\mathds{1}$ is the identity. Therefore, this transformation must preserve commutation relations, which implies that 
\begin{equation}
p'=UpU^{\dagger}\text{, and }q'=UqU^{\dagger},
\end{equation}
with $U$ unitary. Nevertheless, often we have that the operator versions of either $p'$ or $q'$ as functions of $p$ and $q$, even when well defined, have a spectrum that is different from the spectrum of $p$ and $q$. By this fact, we can see that such an unitary may not exist, since unitary maps always preserve the spectrum of any operator. 

The first approach to obtain the correct representation of these canonical transformations was given by Mello and Mochinsky in \cite{mello_nonlinear_1975,moshinsky_canonical_1978,moshinsky_canonical_1979,deenen_canonical_1980}. Their main result was to show that, although one can not find an unitary map taking the operators $p,q$ to the operators $p',q'$, one can find an unitary map that takes operators $J_{1}(p),J_{2}(q)$, into operators ${J'}_{1}(p'),{J'}_{2}(q')$, such that all operators have the same spectrum and degeneracy. This method was refined by Plebanski and Seligman in \cite{plebanski_partial_1980}, where they showed that, essentially, a unitary map $U$ representing the canonical transformation could be built using \textit{partial isometries}.

\begin{definition}
Given any Hilbert space $\mathcal{H}$, a \textit{partial isometry}, is any operator $V$ such that
\begin{equation}
VV^{\dagger}=K,\quad
V^{\dagger}V=L,
\end{equation}
where $K$ and $L$ are projectors into subspaces of $\mathcal{H}$.  If $L=\mathds{1}_{\mathcal{H}}$ we say that $V$ is a \textit{semi-unitary} operator.
\end{definition}

The definition above is equivalent to saying that a partial isometry is any map that acts as an isometry in the subspace orthogonal to its kernel (i.e. its co-image) and is not surjective, while a semi-unitary operator is a partial isometry with trivial kernel. If we have a set of partial isometries $\{V_{\lambda}\}_{\lambda\in\Lambda}$, with $\Lambda$ countable, taking a Hilbert space $\mathcal{H}$ to another Hilbert space $\mathcal{H'}$ such that the union of their images and co-images is $\mathcal{H'}$ and $\mathcal{H}$, respectively, then we can use this set to build a unitary map between $\mathcal{H}$ and $\mathcal{H}'$, as shown in Thm.\ref{thm:unit-part}.

Therefore, since a unitary transformation between hermitian operators with different spectra does not exist, the best option is to see whether a partial isometry is able to do this. According to \cite{plebanski_partial_1980}, for any pair of self-adjoint operators, $X$ and $Y$, such that either their pure-point spectra intersect, or their continuous spectra intersect, or both, there is a partial isometry $V^{X}_{Y}$ between $X$ and $Y$, as given by Thm.\ref{thm:can-par-iso}. Essentially, $V^{X}_{Y}$ is a unitary map between $X_{K}=KXK$ and $Y_{L}=LYL$, where $K$ and $L$ are the projectors into the image and co-image of $V^{X}_{Y}$, respectively.

Thms.\ref{thm:unit-part},\ref{thm:can-par-iso} are the foundations of Alg.\ref{alg:recipe_code}. The idea behind it is the following: if we know beforehand the spectra of the operators $A$ and $B$, we can look for families of operators $\{A_{\mu}\}_{\mu\in M}$ and $\{B_{\nu}\}_{\nu\in N}$ such that the union of the spectra of each family are equal, and the spectra of each operator in the same family has no intersection. Then, we use the canonical partial isometries, which are guaranteed to exist, to find a unitary map between the direct sum of the operators in one family to the direct sum of operators in the other.

This stands to operators with both a pure-point spectrum and a continuous spectrum. If $X$ only has a pure-point spectrum and $Y$ only has a continuous spectrum, there is no canonical partial isometry between $X$ and $Y$, even when their spectra intersect. In this case, we need to define the operator $V^{X}_{Y}$ by
\begin{equation}\label{eq:gen_can_par_iso}
V^{X}_{Y}=\sum_{\mu\in\sigma(X)\cap\sigma(Y)}\ket{\mu}_{X}\bra{\mu}_{Y},
\end{equation}
where $\ket{\mu}_{X}$ is a proper eigenvector of $X$ and $\ket{\mu}_{Y}$ is a generalized eigenvector of $Y$, i.e., it is Dirac normalized. 

The operator $V^{X}_{Y}$ is a well defined map over the Hilbert space spanned by the generalized eigenvectors $\ket{\mu}_{Y}$, which we call $\mathcal{H}_{Y}$. However, ${V^{X}_{Y}}^{\dagger}$ is not well-defined, since the $\ket{\mu}_{Y}$ vectors themselves are not normalizable and do not belong to $\mathcal{H}_{Y}$. Nevertheless, ${V^{X}_{Y}}^{\dagger}$ can be defined as an operator between rigged Hilbert spaces, which are objects that generalize the usual Hilbert space framework to include the Dirac normalized vectors in a rigorous way (cf. Sec.\ref{app_sec:rigged}). In this general framework, we can define $V^{X}_{Y}$ as the \textit{generalized partial isometry} from $Y$ to $X$. Or, in other words:

\begin{definition}{\textit{Generalized partial isometries}}
If $X$ and $Y$ are such that $\sigma(X)=\sigma_{p}(X)$, $\sigma(Y)=\sigma_{c}(Y)$, and $\sigma(X)\cap\sigma(Y)\neq\varnothing$, then we call $V^{X}_{Y}$, given by Eq.(\ref{eq:gen_can_par_iso}), a generalized partial isometry.

A generalized partial isometry must always satisfy
\begin{subequations}
\begin{align}
\tensor[_X]{\mel{\nu}{VV^{\dagger}}{\nu'}}{_X}&=
\begin{cases}
&\delta(\nu-\nu')\text{, if }\nu,\nu'\in\sigma(X)\cap\sigma(Y),\\
&0\text{, otherwise}
\end{cases}
\\
\tensor[_Y]{\mel{\mu}{V^{\dagger}V}{\mu'}}{_Y}&=
\begin{cases}
&\delta_{\mu,\mu'}\text{, if }\mu,\mu'\in\sigma(X)\cap\sigma(Y),\\
&0\text{, otherwise}
\end{cases}
\end{align}
\end{subequations}
\end{definition}

We proved in Sec.\ref{app_sub:code-recip} that families of generalized partial isometries allow us to build unitary representations of transformations between operators with a pure-point spectrum and a continuous spectrum. An informal description of generalized partial isometries already appeared in \cite{plebanski_partial_1980}, but its formalization as maps between rigged Hilbert spaces is another result of our work.

\section{Conclusion}\label{sec:conc}

In conclusion, we provided a grand unification of bosonic codes by which we can obtain the set of logical operators and the set of detectable errors of all known bosonic codes, like the $\mathcal{GKP}$, cat-code and binomial, as well as all codes that are invariant under an arbitrary unitary transformation. In order to do this, first we gave a rigorous definition of unitary-invariant code families, and introduced the family of translation-symmetric codes as an example. We also gave a rigorous definition of what is an ideal code of an unitary-invariant family, and showed that the $\mathcal{GKP}_{\lambda}$ is the unique, ideal translation-symmetric code.

Next, we have provided an algorithm to produce maps between the generators of the unitary stabilizers of two arbitrary unitary-invariant code families. This algorithm is based on our generalization of a method developed to represent classical canonical transformations in quantum mechanics into the rigged Hilbert space formalism. This algorithm also allows us to map the set of logical operators and the set of detectable errors between code families. As an example of application, we used it to show how to map rotation-symmetric codes into translation-symmetric codes, obtaining the first closed expression of the logical Hadamard gate over ideal rotation-symmetric codes that we know of.

We believe that our framework could be used to create codes that are invariant under the dynamics of more general Hamiltonians, and the design of codes tailored to resist specific error models. For example, it would be interesting to study codes that are invariant under the dynamics generated by Gaussian Hamiltonians, since they are quadratic. It would also be interesting to show how the set of logical and detectable error operators of unitary-symmetric codes behave under small, but otherwise arbitrary perturbations of their generators. For example, this might be useful to describe how to change a particular  rotation-symmetric code, like a cat-code, to tailor it to an error model modified by small perturbations. In summary, we believe that our framework will be a useful tool for research on bosonic codes. 

\section*{Acknowledgements}
We acknowledge financial support from the Serrapilheira Institute (grant number Serra-1709-17173), and the Brazilian agencies CNPq (PQ grant No. 305420/2018-6) and FAPERJ (JCN E-26/202.701/2018).

\appendix

\section{Rigged Hilbert Spaces}\label{app_sec:rigged}

In this section we deal with objects called rigged Hilbert spaces. They allow us to rigorously work with the non-normalizable vectors used indiscriminately in quantum mechanics. The primary examples of such vectors are the eigenstates of position and momentum operators, which do not belong to any Hilbert space and yet are still used to write physical, normalizable states. In Sec.\ref{app_sub:nuclear}, we define what are rigged Hilbert spaces, in Sec.\ref{app_sub:spectral} we use it to state the rigged space version of the spectral theorem, and introduce the rigged spaces that we use in our work.

\subsection{Nuclear spaces and Linear Functionals}\label{app_sub:nuclear}

Before defining rigged Hilbert spaces we need to define what is called a nuclear space, and in order to do that we need some preliminary definitions. The definitions here are taken from \cite{de_la_madrid_modino_quantum_2001}.

\begin{definition}{\textit{Comparable norms:}}
Let $\mathcal{D}$ be a vector space, and let $\norm{\ }_{1}$ and $\norm{\ }_{2}$ be two norms in $\mathcal{D}$. We say that these two norms are \textit{comparable} if for all $\psi\in\mathcal{D}$, there is a constant $C>0$ such that
\begin{equation}
\norm{\psi}_{1}\leq C\norm{\psi}_{2}.
\end{equation}
In this case, $\norm{\ }_{1}$ is said to be \textit{weaker} than $\norm{\ }_{2}$, and $\norm{\ }_{2}$ is said to be \textit{stronger} than $\norm{\ }_{2}$.
\end{definition}

\begin{definition}{\textit{Compatible norms:}}
Let $\mathcal{D}$ be a vector space, and let $\norm{\ }_{1}$ and $\norm{\ }_{2}$ be two norms in $\mathcal{D}$. We say that these two norms are \textit{compatible} if any sequence $\{\psi_{n}\}_{n\in\mathbb{N}}$ that converges to $0$ in $\norm{\ }_{1}$, also converges to $0$ in $\norm{\ }_{2}$ and vice-versa.
\end{definition}

\begin{definition}{\textit{Countably Hilbert Space:}}
Let $\mathcal{D}$ be a vector space. We say that $\mathcal{D}$ is a countable Hilbert space, if for any $\psi,\phi\in\mathcal{D}$ there is a sequence of inner products $\{(\phi;\psi)_{p}\}_{p\in\mathcal{N}}$ with
\begin{equation}
(\phi;\psi)_{1}\leq\cdots\leq(\phi;\psi)_{p}\leq\cdots,
\end{equation}
such that the norms
\begin{equation}
\norm{\psi}_{p}:=\sqrt{(\psi,\psi)_{p}},
\end{equation}
are comparable and compatible.
\end{definition}

Informally, a countably Hilbert space is a vector space with an additional structure that makes it a Hilbert space with several notions of "orthogonality" between its elements. This allows us to introduce a special kind of topology on this vector space with interesting properties.

\begin{definition}{\textit{Fréchet Space:}}
Let $\mathcal{D}$ be a countably Hilbert space. The sequence of norms $\{\norm{\ }_{p}\}_{p\in\mathbb{N}}$ generates a topology on $\mathcal{D}$ that we call $\tau_{\mathcal{D}}$. A sequence $\{\psi_{n}\}_{n\in\mathbb{N}}$ converges to $0$ in $\tau_{\mathcal{D}}$ if and only if they converge to $0$ in all norms $\norm{\ }_{p}$ for all $p\in\mathbb{N}$. We call $\mathcal{D}$ a \textit{Fréchet space}, if $\mathcal{D}$ is complete with respect to $\tau_{\mathcal{D}}$.
\end{definition}

 Usually, the subspace $\mathcal{D}$ is defined with respect to an algebra of observables $\mathcal{A}$, as the maximal invariant subspace of this algebra. The topology of $\mathcal{D}$ is generated by a sequence of inner products defined by 
\begin{equation}
(\phi;\psi)_{n}=\mel{\phi}{A^{n}}{\psi}
\end{equation}
for some hermitian operator $A$, which is usually the Hamiltonian of the of problem at hand. When we demand $\mathcal{D}$ built in this to be a Fréchet space, any linear combination of powers of $A$ is becomes a continuous, well defined operator on $\mathcal{D}$. This fact is seen more easily in the following characterization.

\begin{proposition}{\textit{Condition for closure under $\tau_{\mathcal{D}}$}}
Let $\mathcal{D}$ be a countably Hilbert space. Let $\mathcal{D}_{p}$ be the completion of $\mathcal{D}$ under the norm $\norm{\ }_{p}$. Since for all $\psi\in\mathcal{D}$ we have that 
\begin{equation}
\norm{\psi}_{1}\leq\norm{\psi}_{2}\leq\cdots\leq\norm{\psi}_{p}\leq\cdots,
\end{equation}
we have that
\begin{equation}
\mathcal{D}\subset\cdots\subset\mathcal{D}_{p}\subset\cdots\subset\mathcal{D}_{2}\subset\mathcal{D}_{1}.
\end{equation}
Then, $\mathcal{D}$ will be a Fréchet space, if and only if 
\begin{equation}
\mathcal{D}=\bigcap_{p\in\mathbb{N}}\mathcal{D}_{p}.
\end{equation}
\end{proposition}

\begin{proof}
The proof follows directly from the definition of $\tau_{\mathcal{D}}$. See \cite{de_la_madrid_modino_quantum_2001}.
\end{proof}

In our last example, since $(\ ;\ )_{p}$ was defined as the expected value of a power of $A$, each $\mathcal{D}_{n}$ is the domain over which $A^{n}$ is a bounded operator. Therefore $\mathcal{D}$ is the space over which any converging Taylor series over $A$ is well-defined, since the action of all powers of $A$ are well defined in $\mathcal{D}$. There is yet another equivalent definition of $\tau_{\mathcal{D}}$ that allows us to forgo with requiring the norms to be comparable, which is very useful for applications.

\begin{definition}{\textit{Alternate definition of $\tau_{\mathcal{D}}$}}
Let $\mathcal{D}$ be a linear space, and let $P:=\{\norm{\ }_{\mu}\}_{\mu\in M}$ be and indexed collection of compatible norms over $\mathcal{D}$. We say that $P$ is \textit{separating}, if for any $\psi\in\mathcal{D}$ with $\psi\neq0$, there exists a norm $\norm{\ }_{\mu'}\in P$ such that $\norm{\psi}_{\mu'}\neq0$. Let $p\in P$ be a norm such that for every $\psi\in\mathcal{D}$ we have
\begin{equation}
p(\psi)=\underset{M}{\max}\{\norm{\psi}_{\mu}\}_{\mu\in M}.
\end{equation}
Then, the topology $\tau_{\mathcal{D}}$ is the topology generated by taking arbitrary unions and countable intersections of the sets
\begin{equation}
B^{p}_{\epsilon}(\psi):=\{\phi\in\mathcal{D}|p(\phi-\psi)<\epsilon\}.
\end{equation}
\end{definition}

We are now in position to introduce the first step in the construction of a rigged Hilbert space.

\begin{definition}{\textit{Nuclear space:}}
Let $\mathcal{D}$ be a Fréchet Space. Let $\mathcal{D}_{p}$ be the closure of $\mathcal{D}$ with respect to the norm $\norm{\ }_{p}$. Given any $\psi\in\mathcal{D}$ we write is using the symbol $\psi^{[p]}$, when we consider $\psi$ as an element of $\mathcal{D}_{p}$. We say that $\mathcal{D}$ is a \textit{nuclear space} if all of the identity operators
\begin{equation}
{Id}^{m}_{n}:\underset{\psi^{[n]}}{\mathcal{D}\subset\mathcal{D}_{n}}\rightarrow\underset{{Id}^{m}_{n}(\psi^{[n]})=\psi^{[m]}}{\mathcal{D}\subset\mathcal{D}_{m}},
\end{equation}
have a finite trace.
\end{definition}

Essentially, the condition of nuclearity is a technical condition that guarantees that other natural topologies we could choose for $\mathcal{D}$ coincide with the one we have defined. This condition also implies that non-trivial nuclear spaces cannot be Hilbert spaces since, as explained in \cite{de_la_madrid_modino_quantum_2001}, all nuclear Hilbert spaces are finite dimensional. Therefore, nuclear spaces are very "constrained" in a way, such that we can be more loose with the definition of operators over them.

\begin{definition}{\textit{Anti-dual Space:}}
Let $\mathcal{D}$ be a nuclear space. An \textit{anti-linear} functional over $\mathcal{D}$ is map $F:\mathcal{D}\rightarrow\mathbb{C}$ such that
\begin{equation}
F[\alpha\psi_{1}+\beta\psi_{2}]=\alpha^{*}F[\psi_{1}]+\beta^{*}F[\psi_{2}].
\end{equation}
An anti-linear functional $F$ is continuous with respect to $\tau_{\mathcal{D}}$ if and only if there exists a constant $K>0$, and a norm $\norm{\ }_{p}$ such that
\begin{equation}
\abs{F[\psi]}\leq K\norm{\psi}_{p},
\end{equation}
for all $\psi\in\mathcal{D}$. A sequence of anti-linear functionals $\{F_{n}\}_{n\in\mathbb{N}}$ is said to converge to $F$, if and only if the number sequence $\{F_{n}[\psi]\}_{n\in\mathbb{N}}$ converges to $F[\psi]$, for any $\psi\in\mathcal{D}$. This notion of convergence induces a topology over the set of all anti-linear functionals. The closure of set of all continuous anti-linear functionals over $\mathcal{D}$ with respect to this topology is called the \textit{anti-dual space} over $\mathcal{D}$, and is denoted by $\mathcal{D}^{\cross}$. We call the topology over which $\mathcal{D}^{\cross}$ is complete by $\tau_{\mathcal{D}^{\cross}}$ 
\end{definition}

Here, we work with anti-linear functionals instead of linear functionals because, by the Riesz representation theorem, an anti-linear functional $F$ over any Hilbert space $\mathcal{H}$ can be represented as
\begin{equation}\label{eq:dirac-ket-def}
F[\phi]:=\braket{\psi}{F},
\end{equation}
where $\ket{F}$ is seen as an element of $\mathcal{H}$. Therefore, the space of anti-linear functionals is the space where the "kets" of Dirac's notation live. Similarly, the "bra" can be defined by taking the conjugate of Eq.(\ref{eq:dirac-ket-def}) and identifying $\bra{F}$ with the conjugate of $F$, namely $F*$, which is a linear functional.

\begin{definition}{\textit{Constructive definition of $\mathcal{D}^{\cross}$:}}
Let $\mathcal{D}$ be a nuclear space. Let $\mathcal{D}_{p}$ be the Hilbert space formed taking any $\phi,\psi\in\mathcal{D}$ and defining the inner product $(\phi,\psi)_{p}$. The anti-dual space over $\mathcal{D}$ is the space $\mathcal{D}^{\cross}$ given by
\begin{equation}
\mathcal{D}^{\cross}=\bigcup_{p\in\mathbb{N}}\mathcal{D}^{\cross}_{p}, 
\end{equation}
where each $\mathcal{D}^{\cross}_{p}$ is the anti-dual space of $\mathcal{D}_{p}$, which is isomorphic to $\mathcal{D}_{p}$.
\end{definition}

This decomposition of $\mathcal{D}^{\cross}$ as a union of Hilbert spaces $\mathcal{D}_{p}$ makes explicit the fact that $\mathcal{D}^{\cross}$ is not, in general, a Hilbert space itself. There is no general prescription 

\begin{definition}{\textit{Rigged Hilbert Space}}
A \textit{rigged Hilbert space} (abbreviated \textit{RHS}) or \textit{Gelfand-triple} is a triplet of spaces $\mathcal{D}$, $\mathcal{H}$ and $\mathcal{D}^{\cross}$ such that there is a sequence of inner products satisfying
\begin{equation}
(\phi,\psi)_{0}\leq(\phi,\psi)_{1}\leq\cdots\leq(\phi,\psi)_{p}\leq\cdots,
\end{equation}
for which (i:) $\mathcal{D}$ is a nuclear space with respect to $\{(\ ;\ )\}_{p}$ for $p\geq1$; (ii:) The space $\mathcal{H}$ is the Hilbert space $\mathcal{D}_{0}$, and (iii:) $\mathcal{D}^{\cross}$ is the anti-dual of $\mathcal{D}$ such that the map $\iota:\mathcal{D}\rightarrow\mathcal{H}$, given by $\iota(\psi)=\psi$, and its adjoint $\iota^{\cross}:\mathcal{H}^{\cross}\rightarrow\mathcal{D}^{\cross}$ are continuous.
\end{definition}

The maps $\iota$ and $\iota^{\cross}$ work as immersions of $\mathcal{D}$ into $\mathcal{H}$ and of $\mathcal{H}^{\cross}\simeq\mathcal{H}$ into $\mathcal{D}^{\cross}$, respectively. We represent this fact by writing an RHS as a system of inclusions $\mathcal{D}\subset\mathcal{H}\subset\mathcal{D}^{\cross}$. Before proceeding with other constructions, let us summarize what we have seen so far with an example.

\begin{example}{\textit{Schwartz space - Part 1:}}
Consider the algebra of operators generated by the position and momentum operators $q$ and $p$. This algebra has a representation over the Hilbert space of square-integrable functions over $\mathbb{R}$, called $L^{2}(\mathbb{R})$, with inner product given by
\begin{equation}
(f;g)=\int^{\infty}_{-\infty}\dd{x}f^{*}(x)g(x).
\end{equation}
Namely, $q$ and $p$ act as the multiplication and derivative operators over $L^{2}(\mathbb{R})$. However, neither of those are continuous over $L^{2}(\mathbb{R})$ with the topology derived from this inner product, since neither $q$ nor $p$ are bounded in $L^{2}(\mathbb{R})$. 

In order to find a space for which any Taylor series in $q$ and $p$ are continuous, we need to a subspace in which any powers of both $q$ and $p$ are bounded. Therefore, define $\mathcal{D}_{Sch}$ as
\begin{equation}
\mathcal{D}_{Sch}=\bigcap_{n\in\mathbb{N}\backslash{\{0\}}}\mathcal{D}(q^{n})\cap\bigcap_{m\in\mathbb{N}\backslash{\{0\}}}\mathcal{D}(p^{m}),
\end{equation}
where $\mathcal{D}(A^{n})$ is the subspace of $L^{2}(\mathbb{R})$ for which $A^{n}$ is bounded, for any operator $A$. Any element $\omega\in\mathcal{D}_{Sch}$ must then satisfy the conditions
\begin{equation}
\norm{q^{n}p^{m}\omega}=\int^{\infty}_{-\infty}\dd{x}\abs{x^{m}\dv[n]{\omega(x)}{x}(x)}^{2}<\infty,
\end{equation}
for all $m,n\in\mathbb{N}\backslash{\{0\}}$. The space $\mathcal{D}_{Sch}$ is called the \textit{Schwartz space}, and plays a major role in the theory of the Fourier transform and of tempered distributions.

The norms $P=\norm{\psi}_{(n,m)\in\mathbb{N}\cross\mathbb{N}\backslash{\{(0,0)\}}}:=\norm{q^{n}p^{m}\psi}$ are compatible and separating \cite{de_la_madrid_modino_quantum_2001,laan_rigged_2019}. Therefore, $\mathcal{D}_{Sch}$ is in fact complete with respect to the topology $\tau_{\mathcal{D}_{Sch}}$ derived from this set of norms. Now, since an operator $A$ is continuous in $\tau_{\mathcal{D}_{Sch}}$ if and only if it is bounded with respect to one of the norms in $P$, we have that any convergent (normally ordered) Taylor series in classical variables $q,p$ is well defined in $\mathcal{D}_{Sch}$ when we exchange the classical variables by their operator counterparts. It is also easy to see that, since the norm over $L^{2}(\mathbb{R})$ is given by $\norm{q^{0}p^{0}\omega}$ for all $\omega\in\mathcal{D}$, the closure of $\mathcal{D}_{Sch}$ with respect to this norm is exactly $L^{2}(\mathbb{R})$. 
\end{example}

This illustrates our claims about $\mathcal{D}$ being the subset of a Hilbert space over which all observables are well-defined. However, to illustrate the usefulness of an RHS we still need some more machinery, which is the subject of the next subsection.

\subsection{The Nuclear Spectral Theorem and ambiguities in Dirac's notation}\label{app_sub:spectral}

Here we give the most important application of RHS to quantum mechanics in general. What we have is a way to use an RHS as the space containing the "generalized eigenvectors" of an operator with a continuous spectra. This fact is best expressed by the Nuclear spectral theorem, which is our subject for now. But first, a few preliminary definitions

\begin{definition}{\textit{Anti-dual extensions:}}
Let $A$ be a continuous operator over the nuclear space $\mathcal{D}$, of an RHS. We call an operator $A^{\cross}$ the \textit{anti-dual extension of }$A$ if for any anti-linear functional $F\in\mathcal{D}^{\cross}$ we have
\begin{equation}
(A^{\cross}F)[\psi]=F[A\psi],
\end{equation}
for all $\psi\in\mathcal{D}$.
\end{definition}

\begin{definition}{\textit{Generalized eigenvectors:}}
Let $\mathcal{D}\subset\mathcal{H}\subset\mathcal{D}^{\cross}$ be a rigged Hilbert space and $A$ be a self-adjoint, $\tau_{\mathcal{D}}$ continuous operator. Then, a \textit{generalized eigenvector} of $A$, corresponding to a \textit{generalized eigenvalue} $\lambda$ in the spectrum of $A$, is an anti-linear functional $F_{\lambda}\in\mathcal{D}^{\cross}$ such that
\begin{equation}
A^{\cross}F_{\lambda}=\lambda F_{\lambda}.
\end{equation}
\end{definition}

\begin{definition}{\textit{Cyclic operators:}}
An operator $A$ defined over a domain $\mathcal{D}(A)$ of a Hilbert space $\mathcal{H}$ is \textit{cyclic} if there is at least one $\psi\in\mathcal{D}(A)$ such that $\{A^{k}\psi\}_{k\in\mathbb{N}}$ spans $\mathcal{H}$.
\end{definition}

\begin{theorem}{\textit{The Nuclear Spectral Theorem:}}
Let $\mathcal{D}\subset\mathcal{H}\subset\mathcal{D}^{\cross}$ be a rigged Hilbert space and $A$ be a cyclic, self-adjoint, $\tau_{\mathcal{D}}$ continuous operator. Then, for each $\lambda\in\sigma(A)$, there exists;
\begin{enumerate}
    \item An anti-linear functional $F_{\lambda}:=\ket{\lambda}_{A}$ such that $A^{\cross}\ket{\lambda}_{A}=\lambda\ket{\lambda}_{A}$, or in other words,
    \begin{equation}
    \braket{A\psi}{\lambda}_{A}=\mel{\psi}{A^{\cross}}{\lambda}_{A}=\lambda\braket{\psi}{\lambda}_{A},
    \end{equation}
    for all $\psi\in\mathcal{D}$.
    \item A uniquely defined positive measure
    \begin{equation}
    \dd{\mu}(\lambda)=\sum_{\nu\in\sigma_{p}(A)}d_{\nu}\delta(\lambda-\nu)+\rho(\lambda)\dd{\lambda},
    \end{equation}
    such that
    \begin{equation}
    \braket{\phi}{\psi}:=\sum_{\nu\in\sigma_{p}(A)}d_{\nu}\phi^{*}_{\nu}\psi_{\nu}+\int_{\sigma_{c}(A)}\dd{\nu'}\rho(\nu')\phi^{*}(\nu')\psi(\nu'),
    \end{equation}
    where $\psi,\phi\in\mathcal{D}$, with $\braket{\ }{\ }$ being the inner product in $\mathcal{H}$, $\tensor[_A]{\braket{\nu}{\omega}}{}=\omega_{\nu}$ for $\nu\in\sigma_{p}(A)$ and $\tensor[_A]{\braket{\nu}{\omega}}{}=\omega(\nu)$ for $\nu\in\sigma_{c}(A)$, for any $\omega\in\mathcal{D}$.
\end{enumerate}
We also have that, if $f(\lambda)$ is an analytic function over $\sigma(A)$, then \small
\begin{equation}
\mel{\phi}{f(A)}{\psi}=\sum_{\nu\in\sigma_{p}(A)}f(\nu)\phi^{*}_{\nu}\psi_{\nu}+\int_{\sigma_{c}(A)}\dd{\nu'}f(\nu')\phi^{*}(\nu')\phi(\nu').
\end{equation}
\normalsize
\end{theorem}

\begin{proof}
See the appendix to the first chapter of \cite{gelfand_generalized_2016}.
\end{proof}

With this theorem at hand, we are now in position to finish our discussion on the example of the last subsection. 

\begin{example}{\textit{Schwartz space - Part 2:}}

As we have saw in Part 1, the space $\mathcal{D}_{Sch}$ is the nuclear space over which the algebra of position and momentum operators is continuous, and whose closure under usual inner product of quantum mechanics is $L^{2}(\mathbb{R})$. By the Nuclear spectral theorem, the RHS $\mathcal{D}_{Sch}\subset L^{2}(\mathbb{R})\subset\mathcal{D}^{\cross}_{Sch}$ is such that the operators $q$ and $p$ have extensions $q^{\cross}$ and $p^{\cross}$ for which the generalized eigenvalues are well-defined  anti-linear functionals of $\mathcal{D}^{\cross}_{Sch}$.

As an example, the generalized eigenvalue equation for $q$ is 
\begin{equation}
F_{x}[q(\omega)]=xF_{x}[\omega].
\end{equation}
For any particular point $y\in\mathbb{R}$ we have that $q(\omega(y))=y\omega(y)$, which implies that
\begin{equation}
F_{x}[y\omega(y)]=xF_{x}[\omega(y)],
\end{equation}
for any $y\in\mathbb{R}$. 

Using this equation we can look for an integral representation of $F_{x}$, from which we obtain the condition
\begin{equation}
\int^{\infty}_{-\infty}\dd{y}yf_{x}(y)\omega^{*}(y)=x\int^{\infty}_{-\infty}\dd{y}f_{x}(y)\omega^{*}(y),
\end{equation}
for some $f_{x}(y)$. We call $f_{x}(y)$ a \textit{weak solution} of the eigenvalue equation for $q$, and the only $f_{x}(y)$ that satisfies this condition is the Dirac delta function $\delta(x-y)$, from which we find that $F_{x}$ is the anti-linear Dirac functional $\delta^{*}_{x}$, defined  by
$\delta^{*}_{x}[\omega]=\omega^{*}(x)$ for all $\omega\in\mathcal{D}_{Sch}$. 

Therefore, we have that the position eigenkets $\ket{x}_{q}$ are not, in fact, elements of $L^{2}(\mathbb{R})$, and instead are defined by $\ket{x}_{q}:=\delta^{*}_{x}$. By the nuclear spectral theorem we also have that $\braket{\psi}{x}_{q}=\psi^{*}(x)$ for all $\psi\in\mathcal{D}_{Sch}$. Therefore, we can never interpreted this "bra-ket" as an inner product between the "bra" $\bra{\psi}$, and the "ket" $\ket{x}_{q}$, but as the position-representation of $\psi\in\mathcal{D}_{Sch}$.

In fact, by a direct consequence of the nuclear spectral theorem is that any $\psi\in\mathcal{D}_{Sch}$ has the "ket" representation
\begin{equation}
\ket{\psi}=\int^{\infty}_{-\infty}\dd{y}\psi(y)\ket{y}_{q},
\end{equation}
where $\psi(x)=\tensor[_q]{\braket{x}{\psi}}{}$ is the traditional \textit{wave-function} of the state $\psi$. From this representation, we have that the Dirac delta function $\delta(x-x')$ can be rigorously defined as the position-representation of the position eigenstate $\ket{x}_{q}$, or in other words
\begin{equation}
\tensor[_q]{\braket{x'}{x}}{_q}=\delta(x-x').
\end{equation}
\end{example}

In summary. Dirac's notation for generalized eigenvectors is ambiguous because it uses the same symbol, the "bra-ket", to represent both an inner product, the action of a linear functional over a test function, and weak solutions of a eigenvalue problem (i.e such as $\delta(x-x')$). These first two concepts are only equivalent for Hilbert spaces, since they are self-dual, while the third concept is completely unrelated to the last two. Therefore, this conflation leads to weird interpretations about some mathematical objects which may be detrimental to our understanding. To see this, let us study the case of a particular kind of operator.

\begin{definition}{\textit{Generalized projector:}}
Let $\mathcal{D}\subset\mathcal{H}\subset\mathcal{D}^{\cross}$ be a rigged Hilbert space and $A$ be a cyclic, self-adjoint, $\tau_{\mathcal{D}}$ continuous operator with a non-empty continuous spectra $\sigma_{c}(A)$. Let $\ket{\lambda}_{A}\in\mathcal{D}^{\cross}$ be a generalized eigenvector of $A$. The operator $\Pi^{\lambda}_{A}:\mathcal{D}\rightarrow\mathcal{D}^{\cross}$ defined by
\begin{equation}
\Pi^{\lambda}_{A}\ket{\psi}:=\psi(\lambda)\ket{\lambda}_{A},
\end{equation}
is called the \textit{generalized projection} into the \textit{generalized invariant subspace}
\begin{equation}
\mathcal{V}_{\lambda}:=Span(\ket{\lambda}_{A})\subset\mathcal{D}^{\cross}.
\end{equation}
The dimension of the subspace $\mathcal{V}_{\lambda}$ is given by the multiplicity of $\lambda$ in the spectrum of $A$.
\end{definition}

In Dirac's notation, we represent such generalized projections by
\begin{equation}
\Pi^{\lambda}_{A}:=\ketbra{\lambda}{\lambda}_{A}.
\end{equation}
This notation suggests that the domain of $\Pi^{\lambda}_{A}$ can be extended to $\mathcal{D}^{\cross}$. By applying Dirac's notation we would have to define the action of $\Pi^{\lambda}_{A}$ over the generalized eigenvectors $\ket{\lambda'}_{A}$ by
\begin{equation}
\Pi^{\lambda}_{A}\ket{\lambda'}_{A}=\delta(\lambda-\lambda')\ket{\lambda'}_{A}.
\end{equation}
However, the expression $\delta(\lambda-\lambda')\ket{\lambda'}_{A}$ is ill-defined, since we defined $\mathcal{D}^{\cross}$ as a vector space over $\mathcal{C}$, and $\delta(0)$ is not a number.

Therefore, $\Pi^{\lambda}_{A}$ has no extension to $\mathcal{D}^{\cross}$. However, the operator $E_{A}[X]:\mathcal{D}\rightarrow\mathcal{D}^{\cross}$ given by
\begin{equation}
E_{A}[X]:=\int_{X}\dd{\mu}\Pi^{\mu}_{A},
\end{equation}
for any Lebesgue measurable set $X\subset\sigma_{c}(A)$, does have a well-defined extension to $\mathcal{D}^{\cross}$. In fact, these are the projections of the projection-valued measures that are usually employed in quantum mechanics to prove the spectral theorem. Now let us consider some properties of $\Pi^{\lambda}_{A}$.

\begin{proposition}
Let $\Pi^{\lambda}_{A}$ be a generalized projector for the subspace spanned by the generalized eigenvector $\ket{\lambda}_{A}$ of the self-adjoint operator $A$. Then, it is true for all $\phi,\psi\in\mathcal{D}$ that
\begin{enumerate}
    \item $\bra{\psi}\Pi^{\lambda}_{A}=\psi^{*}(\lambda)\tensor[_A]{\bra{\lambda}}{}$, and
    \item $\mel{\phi}{\Pi^{\lambda}_{A}\Pi^{\lambda'}_{A}}{\psi}=\delta(\lambda-\lambda')\mel{\phi}{\Pi^{\lambda}_{A}}{\psi}$.
\end{enumerate}
We call the second property in this list \textit{weak-idempotency}, which is a generalized version of idempotency for operators whose product is not well-defined.
\end{proposition}

\begin{proof}
Since $\phi,\psi\in\mathcal{D}$, just assume that $\Pi^{\lambda}_{A}=\ketbra{\lambda}{\lambda}_{A}$, and compute using Dirac's notation.
\end{proof}

Now we see that for $\mel{\phi}{\Pi^{\lambda}_{A}}{\psi}$ the "bra-ket" as an inner product and the "bra-ket" as a weak solution can both appear as a generalized "mean-value". This is why we must be careful with the definitions of our spaces, and it is also the reason for why we choose to introduce RHS into our framework. These generalized projections play a major role in the decomposition of Hilbert spaces by partial isometries, as we will see in App.\ref{app_sub:part_iso}. Before proceeding to discuss partial isometries, we need one last definition.

\begin{definition}{\textit{Spectral RHS:}}
Let $A$ be self-adjoint operator over a separable Hilbert space $\mathcal{H}$, and let $\sigma(A)\subset\mathbb{R}$ be its spectrum. Its \textit{spectral RHS}, is the RHS $\mathcal{D}_{A}\subset\mathcal{H}_{A}\subset\mathcal{D}^{\cross}_{A}$ such that
\begin{enumerate}
    \item $\mathcal{D}_{A}$ is a nuclear space defined as
    \begin{equation}
    \mathcal{D}_{A}:=\bigcap_{p\in\mathbb{N}\backslash{\{0\}}}\mathcal{D}(A^{p})
    \end{equation}
    with a topology $\tau_{\mathcal{D}_{A}}$ generated by the set of norms $\{\norm{\ }_{p}\}_{p\in\mathbb{N}\backslash{\{0\}}}$ defined by $\norm{\psi}_{p}:=\mel{\psi}{A^{p}}{\psi}$ for all $\psi\in\mathcal{D}_{A}$.
    \item $\mathcal{H}\simeq\mathcal{H}_{A}$, such that the closure of $\mathcal{D}_{A}$ by the norm in $\mathcal{H}_{A}$ is $\mathcal{H}$.
    \item And $\mathcal{D}^{\cross}_{A}$ is the space of all $\tau_{\mathcal{D}_{A}}$-continuous anti-linear functionals over $\mathcal{D}_{A}$.
\end{enumerate}
\end{definition}

Unless state otherwise, we hereby take for granted that all RHS are a spectral RHS. Since we will be always dealing with operators that belong to the algebra of observables generated by $q$ and $p$, we also assume that their powers are also part of the definition of the RHS.

\section{Partial Isometries and Direct-Sum Decomposition}\label{app_sec:part_direct}

In this appendix, we introduce the main tool of our formalism, namely partial isometries. In Sec.\ref{app_sub:part_iso} we define what are partial isometries in general, briefly discuss their properties, and introduce \textit{canonical partial isometries}, which first appeared in \cite{plebanski_partial_1980}. We also propose a generalization of these canonical partial isometries to map subspaces of the anti-linear dual of a spectral nuclear space of one operator to the nuclear space of another, when those operators have spectra of different natures. In Sec.\ref{app_sub:direct_quot} we show how partial isometries can be used to furnish direct-sum decompositions of Hilbert spaces, and show under what conditions these decompositions reveal a quotient structure associated to a particular symmetry.

\subsection{Partial Isometries and generalizations}\label{app_sub:part_iso}

In this subsection we lay out the the definitions of the main objects used in our constructions, the partial isometries. We prove a theorem that guarantees the existence of a partial isometry between the spectral spaces of two self-adjoint operators. These partial isometries, which we call \textit{canonical partial isometries}, have properties that are crucial to building the unitary map used in Alg.\ref{alg:recipe_code}. We also provide an extension of this theorem for the case where the two operators have spectra of different natures (i.e. pure-point vs continuous).

\begin{definition}{\textit{Isometries:}}\label{def:isometries}
Let $\mathcal{H}$ and $\mathcal{H}'$ be separable Hilbert spaces. An isometry $V$ is a map $V:\mathcal{H}\rightarrow\mathcal{H}'$ such that
\begin{equation}
VV^{\dagger}=\mathds{1}_{\mathcal{H}}\quad V^{\dagger}V=\mathds{1}_{\mathcal{H}'}.
\end{equation}
\end{definition}

\begin{definition}{\textit{Partial-isometries:}}\label{def:partial-isometries}
Let $\mathcal{H}$ and $\mathcal{H}'$ be separable Hilbert spaces. A partial isometry of type $(K,L)$ is a map $V:\mathcal{H}\rightarrow\mathcal{H}'$ such that
\begin{equation}
VV^{\dagger}=K\quad V^{\dagger}V=L,
\end{equation}
where $L$ and $K$ are projector over into $\mathcal{H}$ and $\mathcal{H}'$, respectively. A partial isometry $S:\mathcal{H}\rightarrow\mathcal{H}'$ of type $(K,\mathds{1}_{\mathcal{H}})$ is called a \textit{semi-unitary} operator.
\end{definition}

\begin{definition}{\textit{Image and co-image of a partial isometry:}}
Let $V:\mathcal{H}\rightarrow\mathcal{H}'$ be a partial isometry of type $(K,L)$. The subspace stabilized by $K$ is the called the \textit{image} of $V$, or $\Im(V)$, while the subspace stabilized by $L$ is the co-image, or Co-Im$(V)$.
\end{definition}

As stated in Sec.\ref{sec:meth}, a partial isometry maps its co-image isometrically to its image, while a partial isometry with trivial kernel is called a semi-unitary operator. Essentially, a partial isometry becomes a full isometry in the limit where the co-image is equal to the domain, and the image is equal to the range. Now, we define a particular class of partial isometries.

\begin{theorem}{\textit{Canonical partial isometry}}\label{thm:can-par-iso}
Let $X$ and $Y$ be non-degenerate, self-adjoint operators over a Hilbert space $\mathcal{H}$, such that
\begin{enumerate}
    \item We have that $\sigma(X)$ and $\sigma(Y)$ are a countable union of connected, disjoint subsets of $\mathbb{R}$.
    \item For $\sigma(Z)$ with $Z=X,Y$ we have $\sigma(Z)=\sigma_{p}(Z)\cap\sigma_{c}(Z)$ and $\sigma_{p}(Z)\cup\sigma_{c}(Z)=\varnothing$.
    \item We also have that $\sigma_{p}(X,Y):=\sigma_{p}(X)\cup\sigma_{p}(Y)\neq\varnothing$ and $\sigma_{c}(X,Y):=\sigma_{c}(X)\cup\sigma_{c}(Y)\neq\varnothing$.
\end{enumerate} 
Then, the operator
\begin{equation}
V^{X}_{Y}=\sum_{\mu\in \sigma_{p}(X,Y)}\tensor[_X]{\ketbra{\mu}{\mu}}{_Y}+\int_{\sigma_{c}(X,Y)}\dd{\nu}\tensor[_X]{\ketbra{\nu}{\nu}}{_Y},
\end{equation}
is a partial isometry with the following properties; (i:) $V^{X}_{Y}{V^{X}_{Y}}^{\dagger}=K$ and ${V^{X}_{Y}}^{\dagger}V^{X}_{Y}=L$, where
\begin{subequations}
\begin{align}
&K=\sum_{\mu\in\sigma_{p}(X,Y)}\ketbra{\mu}{\mu}_{X}+\int_{\sigma_{c}(X,Y)}\dd{\nu}\ketbra{\nu}{\nu}_{X},\\
&L=\sum_{\mu\in\sigma_{p}(X,Y)}\ketbra{\mu}{\mu}_{Y}+\int_{\sigma_{c}(X,Y)}\dd{\nu}\ketbra{\nu}{\nu}_{Y},
\end{align}
\end{subequations}
and (ii:)
\begin{subequations}
\begin{align}
V^{X}_{Y}Y{V^{X}_{Y}}^{\dagger}&=KXK=X_{K},\label{eq:Y_to_X}\\
{V^{X}_{Y}}^{\dagger}XV^{X}_{Y}&=LYL=Y_{L}.\label{eq:X_to_Y}
\end{align}
\end{subequations}
\end{theorem}

\begin{proof}
To prove (i) just calculate $V^{X}_{Y}{V^{X}_{Y}}^{\dagger}$ and ${V^{X}_{Y}}^{\dagger}V^{X}_{Y}$. To prove (ii), use the assumptions assumptions over $X$ and $Y$ to get the spectral representations
\begin{subequations}
\begin{align}
X&=\sum_{\mu\in\sigma_{p}(X)}\mu\ketbra{\mu}{\mu}_{X}+\int_{\sigma_{c}(X)}\nu\dd{\nu}\ketbra{\nu}{\nu}_{X}\\
Y&=\sum_{\mu\in\sigma_{p}(Y)}\mu\ketbra{\mu}{\mu}_{Y}+\int_{\sigma_{c}(Y)}\nu\dd{\nu}\ketbra{\nu}{\nu}_{Y},
\end{align}
\end{subequations} 
from which we just calculate the right-hand side of Eqs.(\ref{eq:Y_to_X},\ref{eq:X_to_Y}).
\end{proof}

This definition covers the case for operators $X$ and $Y$ whose pure-point spectra intersect, and the case for operators whose continuous spectra intersect. However, for operators $X$ and $Y$, which have only a continuous and a pure-point spectrum respectively, such that those spectra intersect, we can define the following generalization of a partial isometry.

\begin{definition}{\textit{Generalized partial isometries}}
Let $X$ and $Y$ be self-adjoint operators such that $\sigma(X)=\sigma_{p}(X)$, $\sigma(Y)=\sigma_{c}(Y)$, and $\sigma(X)\cap\sigma(Y)\neq\varnothing$. A \textit{generalized canonical partial isometry} is an operator $V^{X}_{Y}:\mathcal{D}^{\cross}_{Y}\rightarrow\mathcal{D}_{X}$, where $\mathcal{D}_{Y}$ is the spectral nuclear space of $Y$ and $\mathcal{D}^{\cross}_{X}$ is the spectral anti-dual space of $X$, such that
\begin{equation}
V^{X}_{Y}=\sum_{\nu\in\sigma(X)\cap\sigma(X)}\tensor[_X]{\ketbra{\nu}{\nu}}{_Y}.
\end{equation}
\end{definition}

\begin{proposition}
Let $X$ and $Y$ be self-adjoint operators such that $\sigma(X)=\sigma_{p}(X)$, $\sigma(Y)=\sigma_{c}(Y)$, and $\sigma(X)\cap\sigma(Y)\neq\varnothing$. If $V^{X}_{Y}$ is a canonical generalized partial isometry between $X$ and $Y$, then
\begin{subequations}
\begin{align}
\mel{\alpha}{VV^{\dagger}}{\beta}&=\delta(0)\mel{\alpha}{K}{\beta}\text{, and }\\
\mel{\phi}{V^{\dagger}V}{\psi}&=\mel{\phi}{L}{\psi},
\end{align}
\end{subequations}
where $\ket{\alpha},\ket{\beta}\in\mathcal{D}_{X}$ and $\ket{\phi},\ket{\psi}\in\mathcal{D}_{Y}$ with ${V^{X}_{Y}}^{\dagger}:\mathcal{D}_{X}\rightarrow\mathcal{D}^{\cross}_{Y}$ given by
\begin{equation}
{V^{X}_{Y}}^{\dagger}=\sum_{\mu\in\sigma(X)\cap\sigma(Y)}\tensor[_Y]{\ketbra{\mu}{\mu}}{_X},
\end{equation}
and
\begin{subequations}
\begin{align}
K&=\sum_{\lambda\in\sigma(X)\cap\sigma(Y)}\Pi^{\lambda}_{X},\\
L&=\sum_{\lambda\in\sigma(X)\cap\sigma(Y)}\Pi^{\lambda}_{Y},
\end{align}
\end{subequations}
such that $\Pi^{\lambda}_{X}$ is a projector and $\Pi^{\lambda}_{Y}$ is a generalized projector.
\end{proposition}

\begin{proof}
Just use the definitions of the operators and expand the the vectors in their respective basis to calculate the left and right-hand sides of each equation.
\end{proof}

Therefore, the properties of a generalized partial isometry are essentially weak versions of the properties of a partial isometry. In summary, partial isometries are generalizations of unitary operators that can be used to map operators with different spectra isometrically in the subspace spanned by their common eigenvectors. In the next subsection, we will see how sets of partial isometries can be used to define unitary maps between Hilbert spaces.

\subsection{Direct-Sum decompositions and Quotient Structures}\label{app_sub:direct_quot}

In this subsection we first show how a set of  partial isometries that covers a space provide a decomposition of this space into a direct-sum of orthogonal subspaces. Next, we show that if the partial isometries in this set are also semi-unitary operators, this direct-sum structure is equivalent to a tensor product decomposition in terms of quotient spaces. These two theorems are later used to proof particular decompositions of the spectral spaces in the operators that appear in Alg.\ref{alg:recipe_code}

\begin{theorem}{\textit{Unitary from partial isometries}}\label{thm:unit-part}
Let $\{V_{\lambda}\}_{\lambda\in\Lambda}$ be an indexed set of partial isometries $V_{\lambda}:\mathcal{H}\rightarrow\mathcal{H}'$ that satisfy
\begin{subequations}
\begin{align}
\sum_{\lambda\in\Lambda}V_{\lambda}V^{\dagger}_{\lambda}&=\sum_{\lambda\in\Lambda}K_{\lambda}=\mathds{1}_{\mathcal{H}'}\text{, and }\\
\sum_{\lambda\in\Lambda}V^{\dagger}_{\lambda}V_{\lambda}&=\sum_{\lambda\in\Lambda}L_{\lambda}=\mathds{1}_{\mathcal{H}},
\end{align}
\end{subequations}
with $K_{\lambda}$ and $L_{\lambda}$ being projectors satisfying $K_{\lambda}K_{\lambda'}=\delta_{\lambda,\lambda'}K_{\lambda'}$ and $L_{\lambda}L_{\lambda'}=\delta_{\lambda,\lambda'}L_{\lambda'}$. Then, (i:) we have that
\begin{equation}
\mathcal{H}'=\bigoplus_{\lambda\in\Lambda}(K_{\lambda}\mathcal{H}')\text{, and }\mathcal{H}=\bigoplus_{\lambda\in\Lambda}(L_{\lambda}\mathcal{H}),
\end{equation}
and we have that (ii:) $U:\mathcal{H}=\bigoplus_{\lambda\in\Lambda}(L_{\lambda}\mathcal{H})\rightarrow\mathcal{H}'$ given by
\begin{equation}
U:=
\begin{bmatrix}
\cdots & V_{\lambda} & \cdots
\end{bmatrix}
,
\end{equation}
is a unitary map.
\end{theorem}

\begin{proof}
Since $\{L_{\lambda}\}_{\lambda\in\Lambda}$ and $\{K_{\lambda}\}_{\lambda\in\Lambda}$ are a set of mutually orthogonal projectors summing to the identity matrix, we have that they factor $\mathcal{H}$ and $\mathcal{H}$' into a direct sum of mutually orthogonal subspaces, respectively, which proves (i). Now, given the form of $U$ in (ii) we can write $U^{\dagger}:\bigoplus_{\lambda\in\Lambda}(K_{\lambda}\mathcal{H}')\rightarrow\mathcal{H}$ as 
\begin{equation}
U^{\dagger}=
\begin{bmatrix}
\vdots\\
V^{\dagger}_{\lambda}\\
\vdots
\end{bmatrix}
.
\end{equation}
Using this form we can calculate $UU^{\dagger}$ and $U^{\dagger}U$ obtaining
\begin{subequations}
\begin{align}
UU^{\dagger}=&\sum_{\lambda\in\Lambda}V_{\lambda}V^{\dagger}_{\lambda}=\sum_{\lambda\in\Lambda}K_{\lambda}=\mathds{1}_{\mathcal{H}'},\\
U^{\dagger}U&=
\begin{bmatrix}
\ddots & 0 & 0\\
0 & L_{\lambda} & 0\\
0 & 0 & \ddots
\end{bmatrix}
=
\mathds{1}_{\mathcal{H}}.
\end{align}
\end{subequations}
\end{proof}

\begin{lemma}\label{lem:semi-cyclic}
Let $\{S_{i}\}_{i=0,...,k-1}$ be a family of semi-unitary operators over some Hilbert space $\mathcal{H}$ such that $S^{\dagger}_{i}S_{i}=\mathbf{1}_{\mathcal{H}}$, and $P_{i}=S_{i}S^{\dagger}_{i}$ satisfy
\begin{subequations}
\begin{align}
P_{i}P_{j}=&\delta_{i,j}P_{j},\label{eq:dis-proj}\\
\sum_{i=0}^{k-1}P_{i}&=\mathbf{1}_{\mathcal{H}}.
\end{align}
\end{subequations}
Then the operator
\begin{equation}
c=\sum_{i=0}^{k-1}S_{i+1}S^{\dagger}_{i},
\end{equation}
where $S_{k}=S_{0}$, is a generator of the cyclic group $C_{k}$ of order $k$, and each $S_{i}$ is an isometry from $\mathcal{H}$ to the quotient space $\mathcal{H}/C_{k}$, i,e $\mathcal{H}$ and $\mathcal{H}/C_{k}$ are isomorphic as Hilbert spaces.
\end{lemma}

\begin{proof}
The condition on Eq.(\ref{eq:dis-proj}) implies that $S^{\dagger}_{l}S_{k}=\delta_{l,k}\mathbf{1}_{\mathcal{H}}$. Using this fact it is easy to show that $c$ generates the cyclic group $C_{k}$, and that $cS_{i}=S_{i+1}$. Let $\mathcal{H}_{i}=P_{i}\mathcal{H}$, than we have that $c\mathcal{H}_{i}=\mathcal{H}_{i+1}$, meaning that all $\mathcal{H}_{i}$ are isomorphic to each other, and also isomorphic to the quotient space $\mathcal{H}/C_{k}$, therefore, each map $S_{i}$ must be an isometry from $\mathcal{H}$ to $\mathcal{H}_{i}\simeq\mathcal{H}/C_{k}$, or in other words
\begin{equation}
S^{\dagger}_{i}S_{i}=\mathbf{1}_{\mathcal{H}},\quad S_{i}S^{\dagger}_{i}=\mathbf{1}_{\mathcal{H}/C_{k}}.
\end{equation}

\end{proof}

\begin{lemma}\label{lem:semi-quo}
Let $A_{k}$ be the group algebra of $C_{k}$, endowed with the inner product $(c^{s},c^{t})=\delta_{s,t\mod{k}}$, and let $\mathcal{K}$ be a finite-dimensional Hilbert space isomorphic to $A_{k}$. Then, the map $\mathbf{S}^{\dagger}:\mathcal{H}\mapsto\mathcal{H}\otimes\mathcal{K}$ defined by
\begin{equation}
\mathbf{S}^{\dagger}=
\begin{bmatrix}
S^{\dagger}_{0}\\
\vdots\\
S^{\dagger}_{k-1}
\end{bmatrix}
,\text{with  }
\mathbf{S}^{\dagger}\ket{\psi}=
\begin{bmatrix}
S^{\dagger}_{0}\ket{\psi}\\
\vdots\\
S^{\dagger}_{k-1}\ket{\psi}
\end{bmatrix}
\end{equation}
is an isometry.\\
\\
\end{lemma}

\begin{proof}
Calculate $\mathbf{S}^{\dagger}\mathbf{S}$ and $\mathbf{S}\mathbf{S}^{\dagger}$ using the properties of products of $S_{i}$ operators.
\end{proof}

These two lemmas shows us under which circumstances the existence of a family of semi-unitary operators reveals the existence of a quotient structure in a Hilbert space.  

\begin{theorem}{\textit{Partial isometries and quotient structures}}
Let $\{V_{i,\alpha}\}_{i=0,...,k;\alpha=0,...,\chi}$ be a set of partial isometries in a Hilbert space $\mathcal{H}$, with $P_{i,\alpha}=V_{i,\alpha}V_{i,\alpha}^{\dagger}$ and $Q_{i,\alpha}=V_{i,\alpha}^{\dagger}V_{i,\alpha}$ such that there are families of semi-unitary operators $\{S_{i}\}_{i=0,...,k}$, and $\{{S'}_{\alpha}\}_{\alpha=0,...,\chi}$ with
\begin{equation}
V_{i,\alpha}=S_{i}{S'}^{\dagger}_{\alpha},
\end{equation}
and
\begin{subequations}
\begin{align}
P_{i,\alpha}=P_{i,\alpha'}\equiv P_{i},&\quad Q_{i,\alpha}=Q_{i',\alpha}\equiv Q_{\alpha},\\
P_{i}P_{i'}=\delta_{i,i'}P_{i},&\quad Q_{\alpha}Q_{\alpha'}=\delta_{\alpha,\alpha'}Q_{\alpha},\\
\sum_{i=0}^{k-1}P_{i}=\mathbf{1}_{\mathcal{H}},&\quad\sum_{\alpha=0}^{\chi-1}Q_{\alpha}=\mathbf{1}_{\mathcal{H}},\\
P_{i}Q_{\alpha}=Q_{\alpha}P_{i}.
\end{align}
\end{subequations}

Then, it is true that there are Hilbert spaces $\mathcal{K}_{1}$ of dimension $\chi$ and $\mathcal{K}_{2}$ of dimension $k$ such that $\mathbf{V}^{\dagger}:\mathcal{H}\otimes\mathcal{K}_{1}\mapsto\mathcal{H}\otimes\mathcal{K}_{2}$ defined by
\begin{equation}
\mathbf{V}^{\dagger}=\mathbf{S}^{\dagger}\cdot\mathbf{S'}=
\begin{bmatrix}
S^{\dagger}_{0}\\
\vdots\\
S^{\dagger}_{k-1}
\end{bmatrix}
\cdot
\begin{bmatrix}
{S'}_{0} & \cdots & {S'}_{\chi-1}
\end{bmatrix}
,
\end{equation}
is an isometry.
\end{theorem}

\begin{proof}
Each family of semi-unitary operators satisfies the conditions of lemma \ref{lem:semi-cyclic} independently, which by lemma \ref{lem:semi-quo} implies the existence of the Hilbert spaces $\mathcal{K}_{1}$ and $\mathcal{K}_{2}$ of the required dimensions, and the existence of isometries $\mathbf{S'}^{\dagger}:\mathcal{H}\mapsto\mathcal{H}\otimes\mathcal{K}_{1}$ and $\mathbf{S}^{\dagger}:\mathcal{H}\mapsto\mathcal{H}\otimes\mathcal{K}_{2}$. Therefore, since $\mathbf{V}^{\dagger}$ is a composition of two isometries it must also be an isometry, proving item 2.
\end{proof}

\section{Algebra Homomorphism and Unitary-Invariant codes}

In this appendix we provide the mathematical justification for the informal algorithm given in Alg.\ref{alg:recipe_code}. In Sec.\ref{app_sub:code-recip}, we prove the main theorem Thm.\ref{thm:main} by employing the results about partial isometries given in Sec.\ref{app_sec:part_direct}. In Sec.\ref{app_sub:trans-log-err}, we apply the main theorem to show that the logical operators and detectable errors that commute with the logical Pauli-$Z$ of an arbitrary unitary-invariant code have a normal form, as well as all logical operators and detectable errors of the associated ideal codes, when those exist.

\subsection{The code transformation algorithm}\label{app_sub:code-recip}

Here we first state the main theorem and its proof in terms of a few lemmas, and in particular Lem.\ref{lem:spec-ext-case3}. The statement and proof of these lemmas compose the majority this section. Essentially, the proof strategy is that we first prove Lem.\ref{lem:spec-ext-case1}, and use it to prove Lem.\ref{lem:spec-ext-case2}, which is then used to prove Lem.\ref{lem:spec-ext-case3}.

\begin{theorem}\label{thm:main}
If $A,B$ are non-degenerate, self-adjoint operators such that their spectra, $\sigma(A)$ and $\sigma(B)$, are a countable union of connected subsets of $\mathbb{R}$ then, there exists two pairs of maps $(\Xi_{A},\Omega_{A})$ and $(\Xi_{B},\Omega_{B})$ such that: $\Omega_{B}\circ\Xi_{A}$ and $\Omega_{A}\circ\Xi_{B}$ are algebra homomorphisms and 
\begin{equation}\label{eq:thm_main}
(\Omega_{B}\circ\Xi_{A})[A]=B,\quad(\Omega_{A}\circ\Xi_{B})[B]=A.
\end{equation}
\end{theorem}

\begin{proof}
By Lem.\ref{lem:spec-ext-case3}, there must exist pairs $(\Xi_{A},\Omega_{A})$ and $(\Xi_{B},\Omega_{B})$ of algebra homomorphisms such that
\begin{equation}
\Xi_{X}[X]=p,\quad\Omega_{X}[p]=X,
\end{equation}
where $X$ is wither $A$ or $B$. Then, we just have to calculate the expressions $(\Omega_{B}\circ\Xi_{A})[A]$ and $(\Omega_{A}\circ\Xi_{B})[B]$ to prove the theorem.
\end{proof}

\begin{lemma}\label{lem:spec-ext-case1}
Let $\mathcal{D}_{Sch}\subset L^{2}(\mathbb{R})\subset\mathcal{D}^{\cross}_{Sch}$ be the RHS where the algebra of observables generated by $q$ and $p$ are $\tau_{\mathcal{D}_{Sch}}$-continuous. Let $A$ be a self-adjoint operator that is an element of this algebra, whose spectrum is continuous and given either by
\begin{enumerate}
    \item a closed, semi-infinite interval $[a,\infty)$ ( alternatively, $(-\infty,a^{'}]$), or
    \item a closed, bounded interval $[b_{1},b_{2}]$.
\end{enumerate}
Then, there is a pair of algebra homomorphisms $(\Xi_{A},\Omega_A)$ such that
\begin{equation}
\Xi_{A}[A]=p,\quad\Omega_{A}[p]=A.
\end{equation}
\end{lemma}

\begin{proof}
First, suppose that $\sigma(A)=[a,\infty)$ (or $(-\infty,a],$). Then, the operator
\begin{equation}
A^{'}:=2a\mathds{1}_{A}-A
\end{equation}
where
\begin{equation}
\mathds{1}_{A}:=\int^{\infty}_{a}\dd{\nu}\ketbra{\nu}{\nu}_{A}
\end{equation}
has $\sigma(A^{'})=(-\infty,a]$ (or $[a,\infty)$). By definition, there are canonical partial isometries $V^{p}_{A}$ and $V^{p}_{A^{'}}$ such that,
\begin{subequations}
\begin{align}
V^{p}_{A}V^{p\dagger}_{A}&=\mathds{1}_{A},\quad V^{p\dagger}_{A}V^{p}_{A}=\int_{\sigma(A)}\dd{\nu}\ketbra{\nu}{\nu}_{p}\text{, and }\\
V^{p}_{A^{'}}V^{p\dagger}_{A^{'}}&=\mathds{1}_{A^{'}},\quad V^{p\dagger}_{A^{'}}V^{p}_{A^{'}}=\int_{\sigma(A^{'})}\dd{\nu}\ketbra{\nu}{\nu}_{p}.
\end{align}
\end{subequations}

These equations imply that $V^{p\dagger}_{A}$ and $V^{p\dagger}_{A^{'}}$ are, in fact, semi-unitary operators satisfying the conditions of Lem.\ref{lem:semi-cyclic} and Lem.\ref{lem:semi-quo}, therefore the map
\begin{equation}
\mathbf{S}^{\dagger}=
\begin{bmatrix}
V^{p}_{A}\\
V^{p}_{A^{'}}
\end{bmatrix}
,
\end{equation}
satisfies
\begin{subequations}
\begin{align}
\mathbf{S}\mathbf{S}^{\dagger}=\int^{\infty}_{a}\dd{\nu}\ketbra{\nu}{\nu}_{p}&+\int^{a}_{-\infty}\dd{\nu}\ketbra{\nu}{\nu}_{p}=\mathds{1}_{p},\\
\mathbf{S}^{\dagger}\mathbf{S}=&\mathds{1}_{A}\oplus\mathds{1}_{A^{'}}=\mathds{1}_{A\oplus A^{'}},
\end{align}
\end{subequations}
and
\begin{equation}
\mathbf{S}p\mathbf{S}^{\dagger}=A\oplus A^{'},\quad\mathbf{S}^{\dagger}(A\oplus A^{'})\mathbf{S}.
\end{equation}

Now, there are always a pair of unique algebra homomorphisms $\iota_{A}:\mathbb{C}[[A]]\rightarrow\mathbb{C}[[A\oplus A^{'}]]$ and $\kappa_{A}:\mathbb{C}[[A]]\rightarrow\mathbb{C}[[A\oplus A^{'}]]$, where $\mathbb{C}[[A]]$ and $\mathcal{C}[[A\oplus A']]$ are the rings of formal power series in the variables "$A$" and "$A\oplus A^{'}$". Then, for any formal power series $f[A]$, we have that 
\begin{equation}
\iota_{A}(f[A])=f[A\oplus A^{'}],\quad\kappa_{A}[f[A\oplus A']]=f[A].
\end{equation}
Therefore, we have that
\begin{subequations}
\begin{align}
\Xi_{A}[A]&:=\mathbf{S}^{\dagger}\iota_{A}[A]\mathbf{S}=V^{p\dagger}_{A}AV^{p}_{A}+V^{p\dagger}_{A^{'}}A^{'}V^{p}_{A^{'}}=p\text{, and }\\
\Omega_{A}[p]&:=\kappa_{A}[\mathbf{S}p\mathbf{S}^{\dagger}]=\kappa_{A}\left(
\begin{bmatrix}
A & 0\\
0 & A^{'}
\end{bmatrix}
\right)
=A
\end{align}
\end{subequations} 
Since isometries are algebra homomorphisms, then $(\Xi_{A},\Omega_{A})$ defined above are also homomorphisms.

Now, assuming that $\sigma(A)=[b_{1},b_{2}]$, define a set of operators $\{A_{n}\}_{n\in\mathbb{Z}}$ where
\begin{equation}
A_{n}=A+n(b_{2}-b_{1})\mathbf{1}_{A}=\int_{\sigma(A)}\dd{\nu}\nu+n(b_{2}-b_{1})\ketbra{\nu}{\nu}_{A}.
\end{equation}
We have that $\sigma(A_{n})=[b_{1}+n(b_{2}-b_{1}),b_{2}+n(b_{2}-b_{1})]$, which implies that the canonical partial isometries $V^{p}_{A_{n}}$ satisfy
\begin{subequations}
\begin{align}
V^{p}_{A_{n}}V^{p\dagger}_{A_{n'}}&=\delta_{n,n^{'}}\mathds{1}_{A^{'}}\text{, and }\\ 
V^{p\dagger}_{A_{n}}V^{p}_{A_{n'}}&=\delta_{n,n^{'}}\int_{\sigma(A_{n^{'}})}\dd{\nu}\ketbra{\nu}{\nu}_{p}=\delta_{n,n^{'}}L_{n}.
\end{align}
\end{subequations}
Therefore, the family $\{V^{p}_{A_{n}}\}_{n\in\mathbb{Z}}$ are a family of semi-unitary operators that satisfy Lem.\ref{lem:semi-cyclic} and Lem.\ref{lem:semi-quo}. 

This fact, allows us to define the Hilbert space $\mathcal{K}$, generated by states $\ket{n}$ with $n\in\mathbb{Z}$, such that we can make the identification $\ket{\nu+n(b_{2}-b_{1})}_{A_{n}}:=\ket{\nu}_{A}\otimes\ket{n}$. This allows to define the operators
\begin{subequations}
\begin{equation}
\mathbf{S}^{\dagger}:=\sum_{n\in\mathbb{Z}}V^{p}_{A_{n}}\bra{n},
\end{equation}
which by Lem.\ref{lem:semi-quo} must be unitary, and the operator
\begin{equation}
J_{A}:=\sum_{n\in\mathbb{Z}}\int_{\sigma(A)}\dd{\nu}(\nu+n(b_{2}-b_{1}))\ketbra{\nu}{\nu}_{A}\otimes\ketbra{n}{n},
\end{equation}
\end{subequations}
which must satisfy
\begin{equation}
\mathbf{S}p\mathbf{S}^{\dagger}=J_{A},\quad\mathbf{S}^{\dagger}J_{A}\mathbf{S}=p.
\end{equation}

As it was done previously, we can also define $\iota_{A}$ and $\kappa_{A}$ as mappings between formal power series over $A$ and $J_{A}$. Then, we have that the maps $(\Xi_{A},\Omega_{A})$ defined by
\begin{subequations}
\begin{align}
\Xi_{A}[A]:=&\mathbf{S}\iota_{A}[A]\mathbf{S}^{\dagger}=\sum_{n\in\mathbb{Z}}V^{p\dagger}_{A_{n}}A_{n}V^{p}_{A_{n}}=p,\\
\Omega_{A}[p]:=&\kappa_{A}[\mathbf{S}^{\dagger}p\mathbf{S}]=\kappa_{A}[J_{A}]=A,
\end{align}
\end{subequations}
are algebra homomorphisms.
\end{proof}

\begin{lemma}\label{lem:spec-ext-case2}
Let $\mathcal{D}_{Sch}\subset L^{2}(\mathbb{R})\subset\mathcal{D}^{\cross}_{Sch}$ be the RHS where the algebra of observables generated by $q$ and $p$ are $\tau_{\mathcal{D}_{Sch}}$-continuous. Let $A$ be a self-adjoint operator that is an element of this algebra, whose spectrum is known to be a pure-point spectrum such that either
\begin{enumerate}
    \item $\sigma_{p}(A)$ finite, or
    \item $\sigma_{p}(A)$ is not finite, but countable and such that    \begin{enumerate}
        \item If $x\in\sigma(A)$ then $x\in[\inf(\sigma(A)),\sup(\sigma(A))]$ or,
        \item $\sigma_{p}(A)$ is not contained in any bounded interval but contained in a semi-infinite interval $[a,\infty)$ (or $(-\infty,a]$) or,
        \item $\sigma_{p}(A)$ is not contained in any bounded nor semi-infinite interval.
\end{enumerate}
Then, there is a pair of algebra homomorphisms $(\Xi_{A},\Omega_A)$ such that
\begin{equation}
\Xi_{A}[A]=p,\quad\Omega_{A}[p]=A.
\end{equation}
\end{enumerate}
\end{lemma}

\begin{proof}
If $\sigma(A)$ is finite, there is an order preserving function $f_{A}:\sigma(A)\rightarrow I_{A}$, with $I_{A}$ being the set of positive integers up to $\abs{\sigma(A)}$, and which we assume as given, such that
\begin{equation}
A=\sum_{i\in I_{A}}f_{A}(i)\ketbra{i}{i}_{A}.
\end{equation}

Let $a\in\mathbb{R}$ with $\max(\sigma(A))<a$. Define the point-wise linear function
\begin{equation}
f^{'}_{A}(x)=\sum^{I_{A}}_{i=1}[f(i)+(f(i+1)-f(i))(x-i)](H_{i}(x)-H_{i+1}(x)),
\end{equation}
where $H_{i}(x):=H(x-i)$ is the Heaviside step function, and we also define $f(I_{A}+1)=a$. Let $\{A_{\nu}\}_{\nu\in[0,1)}$ be a collection of operators such that
\begin{equation}
A_{\nu}=\sum_{i\in I_{A}}f^{'}_{A}(i+\nu)\ketbra{i+\nu}{i+\nu}_{A_{\nu}}.
\end{equation}

We have that the union over all $\sigma(A_{\nu})$ is the closed interval $[\min(\sigma(A)),a]$. Let $V^{p}_{A_{\nu}}$, defined by
\begin{equation}
V^{p}_{A_{\nu}}:=\sum^{I_{A}}_{i=1}\tensor[_{p}]{\ketbra{i+\nu}{i+\nu}}{_{A_{\nu}}},
\end{equation}
be generalized canonical partial isometries. Then, let $\mathcal{D}_{Sch}\subset L^{2}([0,1))\subset\mathcal{D}_{Sch}^{\cross}$ be an RHS such that the multiplication operator
\begin{equation}
\mathbf{M}_{[0,1)}[\psi(x)]=x\psi(x),
\end{equation}
where $\psi\in\mathcal{D}_{Sch}$ and $x\in[0,1)$, has the generalized eigenvectors $\ket{\lambda}_{\mathbf{M}_{[0,1)}}$ such that any element $\ket{\phi}\in L^{2}([0,1))$ can be expanded as
\begin{equation}
\ket{\phi}=\int^{1}_{0}\dd{\mu}\phi(\mu)\ket{\mu}_{\mathbf{M}_{[0,1)}}.
\end{equation}

Now, let us make the identification $\ket{i+\nu}_{A_{\nu}}=\ket{i}_{A}\otimes\ket{\nu}_{\mathbf{M}_{[0,1)}}$, and define the operators
\begin{subequations}
\begin{equation}
\mathbf{S}^{\dagger}=\sum^{I_{A}}_{i=1}\int^{1}_{0}\dd{\lambda}\tensor[_p]{\ket{i+\nu}}{}(\bra{i}_{A}\otimes\bra{\lambda}_{\mathbf{M}_{[0,1)}}),
\end{equation}
and
\begin{equation}
A^{'}:=\int^{1}_{0}\dd{\mu}A_{\mu}\otimes\ketbra{\mu}{\mu}_{\mathbf{M}_{[0,1)}}. 
\end{equation}
\end{subequations}
Then, we can calculate
\begin{subequations}
\begin{align}
\mathbf{S}\mathbf{S}^{\dagger}&=\mathds{1}_{A}\otimes\mathds{1}_{\mathbf{M}_{[0,1)}},\\
\mathbf{S}^{\dagger}\mathbf{S}&=\int^{a}_{\min(\sigma(A))}\dd{\mu}\ketbra{\mu}{\mu}_{p}.
\end{align}
\end{subequations}
Therefore, $\mathbf{S}$ is a semi-unitary operator and $A^{'}$ is an operator with continuous spectrum $\sigma(A^{'})$. Therefore, we can apply Lem.\ref{lem:spec-ext-case1} to $\sigma(A^{'})$ and build the pair $(\Xi_{A},\Omega_{A})$.

If $\sigma(A)$ is not finite, but contained in a bounded interval, then there is an order preserving function $f_{A}:\sigma(A)\rightarrow\mathbb{K}$, where $\mathbb{K}$ is either $\mathbb{N}$ or $\mathbb{Z}$, and which we assume as given, such that
\begin{equation}
A=\sum_{i\in\mathbb{K}}f_{A}(i)\ketbra{i}{i}_{A}.
\end{equation}

Let $a\in\mathbb{R}$ with $\sup(\sigma(A))<a$. Define the point-wise linear function
\begin{equation}
f^{'}_{A}(x)=\sum_{i\in\mathbb{K}}[f(i)+(f(i+1)-f(i))(x-i)](H_{i}(x)-H_{i+1}(x)),
\end{equation}
where $H_{i}(x):=H(x-i)$ is the Heaviside step function. Let $\{A_{\nu}\}_{\nu\in[0,1)}$ be a collection of operators such that
\begin{equation}
A_{\nu}=\sum_{i\in\mathbb{K}}f^{'}_{A}(i+\nu)\ketbra{i+\nu}{i+\nu}_{A_{\nu}}.
\end{equation}

We have that the union over all $\sigma(A_{\nu})$ is the closed interval $[\min(\sigma(A)),a]$. Let $V^{p}_{A_{\nu}}$, defined by
\begin{equation}
V^{p}_{A_{\nu}}:=\sum_{i\in\mathbb{K}}\tensor[_{p}]{\ketbra{i+\nu}{i+\nu}}{_{A_{\nu}}},
\end{equation}
be generalized canonical partial isometries and let $\mathcal{D}_{Sch}\subset L^{2}([0,1))\subset\mathcal{D}_{Sch}^{\cross}$ be the same RHS used in the finite case. Now, let us make the identification $\ket{i+\nu}_{A_{\nu}}=\ket{i}_{A}\otimes\ket{\nu}_{\mathbf{M}_{[0,1)}}$, and define the operators
\begin{subequations}
\begin{equation}
\mathbf{S}^{\dagger}=\sum^{I_{A}}_{i=1}\int^{1}_{0}\dd{\lambda}\tensor[_p]{\ket{i+\nu}}{}(\bra{i}_{A}\otimes\bra{\lambda}_{\mathbf{M}_{[0,1)}}),
\end{equation}
and
\begin{equation}
A^{'}:=\int^{1}_{0}\dd{\mu}A_{\mu}\otimes\ketbra{\mu}{\mu}_{\mathbf{M}_{[0,1)}}. 
\end{equation}
\end{subequations}

Then, we can calculate
\begin{subequations}
\begin{align}
\mathbf{S}\mathbf{S}^{\dagger}&=\mathds{1}_{A}\otimes\mathds{1}_{\mathbf{M}_{[0,1)}},\\
\mathbf{S}^{\dagger}\mathbf{S}&=\int^{a}_{\min(\sigma(A))}\dd{\mu}\ketbra{\mu}{\mu}_{p}.
\end{align}
\end{subequations}
Therefore, $\mathbf{S}$ is a semi-unitary operator and $A^{'}$ is an operator with continuous spectrum $\sigma(A^{'})$. Therefore, we can apply Lem.\ref{lem:spec-ext-case1} to $\sigma(A^{'})$ and build the pair $(\Xi_{A},\Omega_{A})$.

If $\sigma(A)$ is not contained in any bounded interval, but contained in a semi-infinite interval $[a,\infty)$ (or $(-\infty,a]$), there is an order preserving function $f_{A}:\sigma(A)\rightarrow\mathbb{K}$, where $\mathbb{K}$ is either $\mathbb{N}$ or $\mathbb{Z}$, such that
\begin{equation}
A=\sum_{i\in\mathbb{K}}f_{A}(i)\ketbra{i}{i}_{A}.
\end{equation}
Define the point-wise linear function
\begin{equation}
f^{'}_{A}(x)=\sum_{i\in\mathbb{K}}[f(i)+(f(i+1)-f(i))(x-i)](H_{i}(x)-H_{i+1}(x)),
\end{equation}
where $H_{i}(x):=H(x-i)$ is the Heaviside step function. Let $\{A_{\nu}\}_{\nu\in[0,1)}$ be a collection of operators such that
\begin{equation}
A_{\nu}=\sum_{i\in\mathbb{K}}f^{'}_{A}(i+\nu)\ketbra{i+\nu}{i+\nu}_{A_{\nu}}.
\end{equation}

We have that the union over all $\sigma(A_{\nu})$ is the semi-infinite interval $[\min(\sigma(A)),\infty)$. Let $V^{p}_{A_{\nu}}$, defined by
\begin{equation}
V^{p}_{A_{\nu}}:=\sum_{i\in\mathbb{K}}\tensor[_{p}]{\ketbra{i+\nu}{i+\nu}}{_{A_{\nu}}},
\end{equation}
be generalized canonical partial isometries and let $\mathcal{D}_{Sch}\subset L^{2}([0,1))\subset\mathcal{D}_{Sch}^{\cross}$ be the same RHS used in the finite case. Now, let us make the identification $\ket{i+\nu}_{A_{\nu}}=\ket{i}_{A}\otimes\ket{\nu}_{\mathbf{M}_{[0,1)}}$, and define the operators
\begin{subequations}
\begin{equation}
\mathbf{S}^{\dagger}=\sum_{i\in\mathbb{N}}\int^{1}_{0}\dd{\lambda}\tensor[_p]{\ket{i+\nu}}{}(\bra{i}_{A}\otimes\bra{\lambda}_{\mathbf{M}_{[0,1)}}),
\end{equation}
and
\begin{equation}
A^{'}:=\int^{1}_{0}\dd{\mu}A_{\mu}\otimes\ketbra{\mu}{\mu}_{\mathbf{M}_{[0,1)}}. 
\end{equation}
\end{subequations}
Then, we can calculate that
\begin{subequations}
\begin{align}
\mathbf{S}\mathbf{S}^{\dagger}&=\mathds{1}_{A}\otimes\mathds{1}_{\mathbf{M}_{[0,1)}},\\
\mathbf{S}^{\dagger}\mathbf{S}&=\int^{\infty}_{\min(\sigma(A))}\dd{\mu}\ketbra{\mu}{\mu}_{p}.
\end{align}
\end{subequations}
Therefore, $\mathbf{S}$ is a semi-unitary operator and $A^{'}$ is an operator with continuous spectrum $\sigma(A^{'})$. Therefore, we can apply Lem.\ref{lem:spec-ext-case1} to $\sigma(A^{'})$ and build the pair $(\Xi_{A},\Omega_{A})$.

If $\sigma(A)$ is not contained in any bounded interval nor semi-infinite interval, there is an order preserving function $f_{A}:\sigma(A)\rightarrow\mathbb{Z}$, with such that
\begin{equation}
A=\sum_{i\in\mathbb{Z}}f_{A}(i)\ketbra{i}{i}_{A}.
\end{equation}
Define the point-wise linear function
\begin{equation}
f^{'}_{A}(x)=\sum_{i\in\mathbb{Z}}[f(i)+(f(i+1)-f(i))(x-i)](H_{i}(x)-H_{i+1}(x)),
\end{equation}
where $H_{i}(x):=H(x-i)$ is the Heaviside step function. Let $\{A_{\nu}\}_{\nu\in[0,1)}$ be a collection of operators such that
\begin{equation}
A_{\nu}=\sum_{i\in\mathbb{Z}}f^{'}_{A}(i+\nu)\ketbra{i+\nu}{i+\nu}_{A_{\nu}}.
\end{equation}

We have that the union over all $\sigma(A_{\nu})$ is $\mathbb{R}$. Let $V^{p}_{A_{\nu}}$, defined by
\begin{equation}
V^{p}_{A_{\nu}}:=\sum_{i\in\mathbb{Z}}\tensor[_{p}]{\ketbra{i+\nu}{i+\nu}}{_{A_{\nu}}},
\end{equation}
be generalized canonical partial isometries and let $\mathcal{D}_{Sch}\subset L^{2}([0,1))\subset\mathcal{D}_{Sch}^{\cross}$ be the same RHS used in the finite case. Now, let us make the identification $\ket{i+\nu}_{A_{\nu}}=\ket{i}_{A}\otimes\ket{\nu}_{\mathbf{M}_{[0,1)}}$, and define the operators
\begin{subequations}
\begin{equation}
\mathbf{S}^{\dagger}=\sum_{i\in\mathbb{Z}}\int^{1}_{0}\dd{\lambda}\tensor[_p]{\ket{i+\nu}}{}(\bra{i}_{A}\otimes\bra{\lambda}_{\mathbf{M}_{[0,1)}}),
\end{equation}
and
\begin{equation}
J_{A}:=\int^{1}_{0}\dd{\mu}A_{\mu}\otimes\ketbra{\mu}{\mu}_{\mathbf{M}_{[0,1)}}. 
\end{equation}
\end{subequations}

Then, we can calculate that
\begin{subequations}
\begin{align}
\mathbf{S}\mathbf{S}^{\dagger}&=\mathds{1}_{A}\otimes\mathds{1}_{\mathbf{M}_{[0,1)}},\\
\mathbf{S}^{\dagger}\mathbf{S}&=\mathds{1}_{p},
\end{align}
and
\begin{equation}
\mathbf{S}p\mathbf{S}^{\dagger}=J_{A},\quad\mathbf{S}^{\dagger}J_{A}\mathbf{S}=p.
\end{equation}
\end{subequations}

Again, let $\iota_{A}$ and $\kappa_{A}$ be the canonical algebra homomorphisms mapping $A$ to $J_{A}$ and $J_{A}$ to $A$, respectively. Then, we have that the maps $(\Xi_{A},\Omega_{A})$ defined by
\begin{subequations}
\begin{align}
\Xi_{A}[A]:=&\mathbf{S}\iota_{A}[A]\mathbf{S}^{\dagger}=\int^{1}_{0}\dd{\nu}V^{p\dagger}_{A_{\nu}}A_{\nu}V^{p}_{A_{\nu}}=p,\\
\Omega_{A}[p]:=&\kappa_{A}[\mathbf{S}^{\dagger}p\mathbf{S}]=\kappa_{A}[J_{A}]=A,
\end{align}
\end{subequations}
are algebra homomorphisms.
\end{proof}

\begin{lemma}\label{lem:spec-ext-case3}
Let $\mathcal{D}_{Sch}\subset L^{2}(\mathbb{R})\subset\mathcal{D}^{\cross}_{Sch}$ be the RHS where the algebra of observables generated by $q$ and $p$ are $\tau_{\mathcal{D}_{Sch}}$-continuous. Let $A$ be a self-adjoint operator that is an element of this algebra, whose spectra is a countable union of connected subsets of $\mathbb{R}$, such that $\sigma_{p}(A)$ is a discrete set, and $\sigma_{c}(A)$ is a closed set. Then, there is a pair of algebra homomorphisms $(\Xi_{A},\Omega_A)$ such that
\begin{equation}
\Xi_{A}[A]=p,\quad\Omega_{A}[p]=A.
\end{equation}
\end{lemma}

\begin{proof}
If $A$ is as above, then either
\begin{equation}
\sigma(A)=\sigma(A_{-\omega})\cup\left(\bigcup_{i\in\mathbb{Z}}\sigma(A_{i})\right)\cup\sigma(A_{\omega}),
\end{equation}
such that $A_{i}$, $A_{-\omega}$ and $A_{\omega}$ can be either the null operator, or an operator whose spectrum,
\begin{enumerate}
    \item in the case of $A_{i}$ with $i\in\mathbb{Z}$, is bounded by $[\inf(\sigma_{p}(A)),\max(\sigma_{c}(A))]$ such that either
    \begin{enumerate}
    \item $\sigma_{p}(A_{i})=\varnothing$ and $\sigma_{c}(A_{i})\neq\varnothing$, or 
    \item $\sigma_{p}(A_{i})\neq\varnothing$ and $\sigma_{c}(A_{i})=\varnothing$, or
    \item $\sigma_{p}(A_{i})\neq\varnothing$ and $\sigma_{c}(A_{i})\neq\varnothing$, with $\sup(\sigma_{p}(A_{i}))\leq\min(\sigma_{c}(A_{i}))$,
\end{enumerate}
    \item while in the case of $A_{-\omega}$ the previous conditions apply, but with $\sigma(A_{-\omega})$ not belonging to any bounded interval, only to $(-\infty,\max(\sigma_{c}(A_{-\omega}))$, 
    \item and in the case of $A_{\omega}$ the same conditions also apply, but with $\sigma(A_{\omega})$ not belonging to any bounded interval, only to $[\inf(\sigma_{p}(A_{\omega}),\infty)$.
\end{enumerate}

Under these conditions, each operator $A$ has the direct product decomposition
\begin{equation}
A=A_{-\omega}\oplus\left(\bigoplus_{i\in\mathbb{Z}}A_{i}\right)\oplus A_{\omega},
\end{equation}
where each $A_{i}$, $A_{-\omega}$ and $A_{\omega}$ can be further decomposed as a direct-sum of a discrete part $A^{d}$ and a continuous part $A^{c}$, giving us the decomposition
\begin{equation}
A=(A^{d}_{-\omega}\oplus A^{c}_{-\omega})\oplus\left(\bigoplus_{i\in\mathbb{Z}}(A^{p}_{i}\oplus A^{c}_{i})\right)\oplus(A^{d}_{\omega}\oplus A^{c}_{\omega}).
\end{equation}

Let $S_{d}$ be the set containing the indices for which $A_{\nu}$ have non-empty $\sigma_{p}(A_{\nu})$, where $\nu$ is either in $\mathbb{Z}$, $\omega$, or $-\omega$. By construction $S_{d}$ must be countable, and for each $s\in S_{d}$ we have an operator $A_{s}=A^{d}_{s}\oplus A^{c}_{s}$, for which we can apply Lem.\ref{lem:spec-ext-case2} to $A^{d}_{s}$ choosing $a=\min(\sigma_{c}(A^{c}_{s})$. From this, we obtain a set of operators $\{A^{d}_{s,\lambda}\}_{(s,\lambda)\in S_{d}\cross[0,1)}$, from which we can construct a set of partial isometries $\mathbf{V}_{s}$, and a set of operators $A^{d'}_{s}$ such that $\sigma(A^{d'}_{s}\oplus A^{c}_{s})=[\inf(\sigma_{p}(A_{s})),\max(\sigma_{c}(A_{s}))]$, if $s\in\mathbb{Z}$. In $s=\omega$ then $\sigma(A^{d'}_{s}\oplus A^{c}_{s})=[\inf(\sigma_{p}(A_{s})),\infty)$, and if $s=-\omega$ then $\sigma(A^{d'}_{s}\oplus A^{c}_{s})=(-\infty,\max(\sigma_{c}(A_{s}))]$.

This construction allows us to build a partial isometry
\begin{equation}
\mathbf{W}:=
\begin{bmatrix}
\vdots\\
\mathbf{V}_{i}\\
\vdots
\end{bmatrix}
\end{equation}
where $i\in S_{d}$ that relates $A$ and the operator $A^{'}$, which has the same decomposition as $A$, but with the operator $A^{d'}_{s}$ in the place of $A^{d}_{s}$ for all $s\in S_{d}$. The spectrum of $A^{'}$ must, therefore, be at most a countable union of intervals closed in the topology of $\mathbb{R}$. 

Therefore, the difference $\mathbb{R}\backslash{\sigma(A^{'})}$ must satisfy
\begin{equation}
\mathbb{R}\backslash{\sigma(A^{'})}=N_{-\omega}\cup\left(\bigcup_{i\in\mathbb{Z}}N_{i}\right)\cup N_{\omega},
\end{equation}
where each $N_{i}$, or $N_{\omega}$, or $N_{-\omega}$ is empty, or an interval, where $N_{i}$ must be bounded, while $N_{\omega}$ and $N_{-\omega}$ must be semi-infinite. We can now chose a particular $A^{'}_{i}$ in the direct-sum decomposition of $A^{'}$, and produce a set of operators $\{A^{'}_{n}\}_{n\in I}$ where $I$ is at most $\mathbb{Z}$, and $A^{'}_{0}=A^{'}$, such that 
\begin{equation}
\bigcup_{n\in I}\sigma(A^{'}_{n})=\mathbb{R}.
\end{equation}
Then, Lem.\ref{lem:spec-ext-case1} gives us a way to build the algebra homomorphism $\Xi_{A}$ and $\Omega_{A}$ required by this lemma.
\end{proof}

\subsection{Transformation of Logical and Error Operators}\label{app_sub:trans-log-err}

Here we show how the main theorem can be applied to obtain a normal form of the logical operators and detectable error operators of unitary-invariant codes. First, we state a corollary of Lem.\ref{lem:spec-ext-case3}, that allows us to put the unitary and hermitian operators constructed in Alg.\ref{alg:recipe_code} in a normal form. Next, we give a Hilbert space representation of the action over operators of the canonical algebra homomorphism $\kappa_{A}$. We also extended this action from operators to states. These representations are used to construct the normal forms we claim to exist in Thm.\ref{thm:log-op}.

\begin{corollary}\label{cor:norm-form}
Let $\mathcal{D}_{Sch}\subset L^{2}(\mathbb{R})\subset\mathcal{D}^{\cross}_{Sch}$ be the RHS where the algebra of observables generated by $q$ and $p$ are $\tau_{\mathcal{D}_{Sch}}$-continuous. Let $A$ be a self-adjoint operator that is an element of this algebra, whose spectra is a countable union of connected subsets of $\mathbb{R}$, such that $\sigma_{p}(A)$ is a discrete set, and $\sigma_{c}(A)$ is a closed set. Then, there is a collection of operators $\{A_{(s,\lambda)}\}_{(s,\lambda)\in S\cross[0,1)}$ where $S$ is at most countable, and can be chosen as to contain $0$, such that $A_{s_{0},0}=A$, $U$ is an operator of the form
\begin{equation}
U^{\dagger}=\sum_{s\in S}\int^{1}_{0}\dd{\lambda}V^{p}_{A_{(s,\lambda)}}\bra{s}\otimes\bra{\lambda},
\end{equation}
and $J_{A}$ is a self-adjoint operator of the form
\begin{equation}
J_{A}=\sum_{s\in S}\int^{1}_{0}\dd{\lambda}A_{(s,\lambda)}\otimes\ketbra{s}{s}\otimes\ketbra{\lambda}{\lambda},
\end{equation}
where $V^{p}_{A_{(s,\lambda)}}$ are (possibly generalized) canonical partial isometries, and there is an RHS $\mathcal{D}\subset L^{2}([0,1))\subset\mathcal{D}^{\cross}$ such that $\ket{\lambda}$ is an anti-linear functional in $\mathcal{D}^{\cross}$ satisfying
\begin{equation}
\ket{\lambda}[\omega]:=\omega^{*}(\lambda),
\end{equation}
for any $\omega\in\mathcal{D}$ and $\lambda\in[0,1)$, and where $\{\ket{s}|s\in S\}$ span a Hilbert space $\mathcal{K}$.
\end{corollary}

\begin{proposition}{\textit{Representations of }$\kappa_{A}$:}\label{prop:repr-of-kappa}
Let $A$ be an operator satisfying the conditions of Thm.\ref{thm:main}. The unique algebra homomorphism $\kappa_{A}:\mathbb{C}[[J_{A}]]\rightarrow\mathbb{C}[[A]]$ can be represented by the map
\begin{equation}
\kappa_{A}[O]:=\underset{\epsilon\rightarrow0^{+}}{\lim}\mel{\kappa^{\epsilon}_{A}}{O}{\kappa^{\epsilon}_{A}},
\end{equation}
where $O$ is any operator of the form $UO(q,p)U^{\dagger}$ with $U$ having the normal form given in Cor.\ref{cor:norm-form}, and the state $\ket{\kappa^{\epsilon}_{A}}$ is given by
\begin{equation}
\ket{\kappa^{\epsilon}_{A}}:=\ket{0}\otimes\int^{1}_{0}\dd{\nu}\eta_{\epsilon}(\nu)\ket{\nu},
\end{equation}
where $\eta_{\epsilon}$ satisfies
\begin{equation}
\underset{\epsilon\rightarrow0^{+}}{\lim}\abs{\eta_{\epsilon}}^{2}=\delta(0).
\end{equation}
\end{proposition}

\begin{proof}
We need to prove that the map we just defined is the same as $\kappa_{A}$, when acting over formal power series in the operator $J_{A}$. By the canonical form of Cor.\ref{cor:norm-form}, any power of $J_{A}$ is of the form
\begin{equation}
(J_{A})^{n}=\sum_{s\in S}\int^{1}_{0}\dd{\lambda}(A_{(s,\lambda)})^{n}\otimes\ketbra{s}{s}\otimes\ketbra{\lambda}{\lambda}.
\end{equation}
Therefore, for any formal power series we have
\begin{equation}
\mel{\kappa^{\epsilon}_{A}}{\sum^{\infty}_{n=1}\alpha_{n}(J_{A})^{n}}{\kappa^{\epsilon}_{A}}=\int^{1}_{0}\dd{\lambda}\left(\sum^{\infty}_{n=1}\alpha_{n}(A(\lambda))^{n}\right)\abs{\eta_{\epsilon}(\lambda)}^{2},
\end{equation}
which implies
\begin{equation}
\underset{\epsilon\rightarrow0^{+}}{\lim}\mel{\kappa^{\epsilon}_{A}}{\sum^{\infty}_{n=1}\alpha_{n}(J_{A})^{n}}{\kappa^{\epsilon}_{A}}=\sum^{\infty}_{n=1}\alpha_{n}A^{n}.
\end{equation}
\end{proof}

\begin{definition}
Let $A$ be a self-adjoint operator satisfying the conditions of Thm.\ref{thm:main}. Also, let $\ket{\psi}$ be any state of the form
\begin{equation}
\ket{\psi}=\int^{\infty}_{-\infty}\dd{z}\psi(z)\ket{z}_{p}.
\end{equation}
The action of $\Omega_{A}$ over $\ket{\psi}$ is defined by 
\begin{equation}
\Omega_{A}[\ket{\psi}]:=\bra{\kappa_{A}}U\ket{\psi},
\end{equation}
where $\bra{\kappa_{A}}=\bra{0}\otimes\bra{0}$ and $U$ has the normal form in Cor.\ref{cor:norm-form}.
\end{definition}

\begin{proposition}\label{prop:weyl-proj}
Let $\mathcal{D}_{Sch}\subset L^{2}(\mathbb{R})\subset\mathcal{D}^{\cross}_{Sch}$ be the RHS where the algebra of observables generated by $q$ and $p$ are $\tau_{\mathcal{D}_{Sch}}$-continuous, and let $A$ be a self-adjoint operator that is an element of this algebra. If the spectrum of $A$ is a countable union of connected subsets of $\mathbb{R}$, then for every $s,t\in\mathbb{R}$, the operator $\Omega_{A}[e^{itq}]$ satisfies
\begin{equation}
\Omega_{A}[e^{itq}]e^{isA}=e^{-ist}e^{isA}\Omega_{A}[e^{itq}].
\end{equation}
\end{proposition}

\begin{proof}
By Thm.\ref{thm:main}, there is a family of operators $\{A_{n,\lambda}\}_{(n,\lambda)\in S\cross[0,1)}$ with $0\in S$ and $S$ at most countable, an unitary $U$, and a self-adjoint $J_{A}$ with the normal forms given in Cor.\ref{cor:norm-form}. Now, we can take the Weyl commutation relations
\begin{equation}
e^{itq}e^{isp}=e^{-ist}e^{isp}e^{itq},
\end{equation}
and multiply from the left by $U$ and from the right by $U^{\dagger}$ to obtain the commutation relations
\begin{equation}
e^{it(UqU^{\dagger})}e^{isJ_{A}}=e^{-ist}e^{isJ_{A}}e^{it(UqU^{\dagger})}.
\end{equation}

This equation is can be expanded and we can collect coefficients obtaining a set of equations
\begin{equation}
\{B(t)\}^{(n',\lambda')}_{(n,\lambda)}e^{isA_{(n',\lambda')}}=e^{-ist}e^{isA_{(n,\lambda)}}\{B(t)\}^{(n',\lambda')}_{(n,\lambda)},
\end{equation}
where $\{B(t)\}^{(n',\lambda')}_{(n,\lambda)}=[U^{\dagger}e^{itq}U]^{(n',\lambda')}_{(n,\lambda)}$. Therefore, in particular, we have that
\begin{equation}
\{B(t)\}^{(0,0)}_{0,0}e^{isA_{(0,0)}}=e^{-ist}e^{isA_{(0,0)}}\{B(t)\}^{(0,0)}_{0,0}.
\end{equation}
now, just note that $\Omega_{A}[e^{itq}]=\kappa_{A}[U^{\dagger}e^{itq}U]=\{B(t)\}^{(0,0)}_{0,0}$. Proving our statement.
\end{proof}

\begin{theorem}\label{thm:log-op}
Let $\mathcal{D}_{Sch}\subset L^{2}(\mathbb{R})\subset\mathcal{D}^{\cross}_{Sch}$ be the RHS where the algebra of observables generated by $q$ and $p$ are $\tau_{\mathcal{D}_{Sch}}$-continuous, and let $A$ be a self-adjoint operator that is an element of this algebra. Let $\mathcal{S}_{U_{2s}}$ be an unitary-invariant code stabilized by $U_{2s}=e^{i2sA}$. Then it is true that 
\begin{enumerate}
    \item the operators $\Omega_{A}[e^{isp}]$, $\Omega_{A}[e^{i(s^{2}/2\pi)p^{2}}]$ and $\Omega_{A}[e^{i(s^{4}/4\pi^{3})p^{4}}]$ are, respectively, the logical $\Bar{Z}$,$\Bar{S}$ and $\Bar{T}$ single-qubit gates for any $\mathcal{S}_{U_{2s}}$ code, and
    \item the operator $\Omega_{A}[e^{i(-2\pi/s)q}]$ is a stabilizer of the ideal $\mathcal{S}_{U_{2s}}$ code, and
    \item the operators 
    \begin{equation}
    \Omega_{A}[e^{-i(\pi/s)q}]\text{, and }\Omega_{A}[e^{i\pi/2((s^{2}p^{2}/\pi)+(\pi q^{2}/s^{2}))}]
    \end{equation}
    are, respectively the logical $\Bar{X}$ and $\Bar{H}$ single-qubit gates for the ideal $\mathcal{S}_{U_{2s}}$ code.
\end{enumerate}
\end{theorem}

\begin{proof}
Let $\mathcal{C}^{p}_{s}:=\{\ket{k\pi/s}_{p}|k\in\mathbb{Z}\}$ be the generalized subspace spanned by the generalized eigenstates of $p$ that is stabilized by $e^{i2sp}$. By construction, we known that $Ue^{i2sp}U^{\dagger}=e^{i2sJ_{A}}$, where $J_{A}$ and $U$ have the normal forms given in Cor.\ref{cor:norm-form}. Since $e^{i2sp}$ acts as the identity operator over $\mathcal{C}^{p}_{s}$, and $U$ is unitary, $e^{i2sJ_{A}}$ must also act as the identity operator over $\mathcal{C}^{J_{A}}_{s}=U\mathcal{C}^{p}_{s}$. This allows us to choose a basis for $\mathcal{C}^{J_{A}}_{s}$ as follows. For every $k\pi/s$ in the spectrum of $p$, there must be exactly one operator $A_{(l_{k},\mu_{k})}$ whose spectrum contains $k\pi/s$, such that different values of $k$ might be associated to the same operator. Then, define the (possibly generalized) eigenvectors
\begin{equation}
\ket{k,l_{k},\mu_{k}}:=\ket{k(\pi/s)}_{A_{(l_{k},\mu_{k})}}\otimes\ket{l_{k}}\otimes\ket{\mu_{k}},
\end{equation}
with $k\in\mathbb{Z}$, $l_{k}\in S$ and $\mu_{k}\in[0,1)$.The set of all of these eigenvectors is a basis for $\mathcal{C}^{J_{A}}_{s}$. 

Now, let $\ket{\Theta}_{p}$ be a primitive state that generates a translation-symmetric code with stabilizer $T^{q}_{2s}$. The state $\ket{\Theta}_{A}=\Omega_{A}[\ket{\Theta}_{p}]$ will only be a primitive state for a $\mathcal{S}_{U_{2s}}$ code if and only if $\mathcal{C}^{J_{A}}_{s}$ contains basis states of the form $\ket{n,0,0}$ with $n\in N$, such that $N$ contains at least one even and one odd integer, and these states are present in the expansion of $U\ket{\Theta}_{p}$ in the $\mathcal{C}^{J_{A}}_{s}$ basis.

With this basis at hand, notice that
\begin{subequations}
\begin{align}
Ue^{isp}U^{\dagger}\ket{k,l_{k},\mu_{k}}&=(-1)^{k}\ket{k,l_{k},\mu_{k}},\\
Ue^{i\left(\frac{s^{2}}{2\pi}\right)p^{2}}U^{\dagger}\ket{k,l_{k},\mu_{k}}&=e^{i\frac{\pi}{2}(k(\textup{mod}2))}\ket{k,l_{k},\mu_{k}},\\
Ue^{i\left(\frac{s^{4}}{4\pi^{3}}\right)p^{4}}U^{\dagger}\ket{k,l_{k},\mu_{k}}&=e^{i\frac{\pi}{4}(k(\textup{mod}2))}\ket{k,l_{k},\mu_{k}}.
\end{align}
\end{subequations}
Therefore, for any primitive state $\ket{\Theta}_{A}$ of the form described above, we have that the operators $\Omega_{A}[e^{isp}]$, $\Omega_{A}[e^{i(s^{2}/2\pi)p^{2}}]$ and $\Omega_{A}[e^{i(s^{4}/4\pi^{3})p^{4}}]$ act as the logical $\Bar{Z}$,$\Bar{S}$ and $\Bar{T}$ single-qubit gates over the codewords
\begin{subequations}
\begin{align}
\ket{j}^{\Theta_{A}}_{\mathcal{S}_{U_{2s}}}&=\frac{\Pi^{j}_{U_{2s}}\ket{\Theta_{A}}}{\sqrt{\mel{\Theta_{A}}{\Pi^{j}_{U_{2s}}}{\Theta_{A}}}}\text{, with}\\
\Pi^{j}_{U_{2s}}&=\sum_{m\in\mathbb{Z}}((-1)^{j}U_{2s})^{m}.
\end{align}
\end{subequations}

Next, remember that the $\mathcal{GKP}_{-s/\sqrt{\pi}}$ codewords are given by
\begin{equation}
\ket{j}_{\mathcal{GKP}_{-s\sqrt{\pi}}}:=\sum_{m\in\mathcal{Z}}\ket{(2m+j)\pi/s}_{p}
\end{equation}
It is not hard to see that we can use the description of our basis to write in the form $\ket{j}^{*}_{e^{isJ_{A}}}=U\ket{j}_{\mathcal{GKP}_{-s\sqrt{\pi}}}$
\begin{subequations}
\begin{equation}
\ket{j}^{*}_{e^{isJ_{A}}}=\sum_{m\in\mathbb{Z}}\ket{(2m+j),l_{2m+j},\mu_{2m+j}}.
\end{equation}
Again, in order to have a well-defined ideal $\mathcal{S}^{*}_{U_{2s}}$ code, we need to have states of the form $\ket{2m_{j}+j,0,0}:=\ket{2m_{j}+j}$ where $m_{j}\in M_{j}$ are integers that may depend on $j$, and $M_{j}$ both are countably infinite. Then, we can write the previous sum as
\begin{equation}
\sum_{m_{0}\in M_{j}}\ket{2m_{j}+j}+\sum_{n\in\mathbb{Z}\backslash{M_{j}}}\ket{(2n+j),l_{2n+j},\mu_{2n+j}}.
\end{equation}
\end{subequations}

\begin{equation}
Ue^{-i\frac{n\pi}{s}q}U^{\dagger}\ket{k,l_{k},\mu_{k}}=\ket{k+n,l_{k+n},\mu_{k+n}}.
\end{equation}
This equation implies that $Ue^{-i\frac{n\pi}{s}q}U^{\dagger}$ changes $k$ in $\ket{k,l_{k},\mu_{k}}$ by cycling through members of the family $\{A_{l_{k},\mu_{k}}\}$, in order to find the operators that have $k+n$ as an eigenvalue. Therefore, since the vector $\ket{0}^{*}_{e^{isJ_{A}}}$ contains all $\ket{k,l_{k},\mu_{k}}$ with $k$ even, and $\ket{1}^{*}_{e^{isJ_{A}}}$ contains all $\ket{k,l_{k},\mu_{k}}$ with $k$, we have that $Ue^{-i\frac{2\pi}{s}q}U^{\dagger}$ acts as the identity and $Ue^{-i\frac{n\pi}{s}q}U^{\dagger}$ acts as the Pauli-X over $\ket{j}^{*}_{e^{isJ_{A}}}$.

Now, by the Weyl commutation relations we have that
\begin{subequations}
\begin{align}
e^{-i\frac{2\pi}{s}q}e^{i2sp}=e^{i2sp}e^{-i\frac{2\pi}{s}q}\text{, and }\\
e^{-i\frac{\pi}{s}q}e^{isp}=-e^{isp}e^{-i\frac{\pi}{s}q}.
\end{align}
\end{subequations}
This fact, when taken with the last considerations and with Prop.\ref{prop:weyl-proj}, shows that $\Omega_{A}[e^{-i\frac{2\pi}{s}q}]$ is a stabilizer of the ideal $\mathcal{S}_{U_{2s}}$ code, and that $\Omega_{A}[e^{-i\frac{\pi}{s}q}]$ is its logical Pauli-X.

Lastly, to show that $\Omega_{A}[e^{i\pi/2((s^{2}p^{2}/\pi)+(\pi q^{2}/s^{2}))}]$ is a logical Hadarmard gate for the ideal $\mathcal{S}_{U_{2s}}$ code, notice that the fact that $e^{i\pi/2((s^{2}p^{2}/\pi)+(\pi q^{2}/s^{2}))}$ is the Hadamard gate for $\mathcal{GKP}_{-s/\sqrt{\pi}}$ implies
\begin{multline}
Ue^{i\pi/2((s^{2}p^{2}/\pi)+(\pi q^{2}/s^{2}))}U^{\dagger}=\\
=e^{-i(\pi/s)UqU^{\dagger}}(Ue^{i\pi/2((s^{2}p^{2}/\pi)+(\pi q^{2}/s^{2}))}U^{\dagger})e^{isJ_{A}}.
\end{multline}
Then, we can take the matrix elements of this equation with relation to the eigenvectors $\ket{k,l_{k},\mu_{k}}$, to show that $Ue^{i\pi/2((s^{2}p^{2}/\pi)+(\pi q^{2}/s^{2}))}U^{\dagger}$ acts as a Hadarmard gate over the states $\ket{j}^{*}_{e^{isJ_{A}}}$, implying that $\Omega_{A}[e^{i\pi/2((s^{2}p^{2}/\pi)+(\pi q^{2}/s^{2}))}]$ acts as a Hadamard over the ideal $\mathcal{S}_{U_{2s}}$ code.
\end{proof}

\end{document}